%% file: Secure_Multiplexing.tex
\documentclass[journal]{IEEEtran}

\usepackage{graphicx}
\usepackage{amsmath}
\usepackage{amsfonts}
\usepackage{amssymb}
\usepackage{dsfont}
\usepackage{bm}
\usepackage{amsthm}
\usepackage{mathrsfs}
\usepackage{newlfont}
\usepackage{float}
\usepackage{hyperref}
\usepackage{algorithm}
\usepackage{algorithmic}
\usepackage{enumerate}
\usepackage{chngcntr}
\usepackage{mathtools}
\usepackage{breqn}
\usepackage{stmaryrd}

\usepackage{subfigure}
\usepackage{acronym}
\usepackage{cite}

\usepackage[utf8x]{inputenc} 
\hypersetup{
    colorlinks=true, 
    breaklinks=true, 
    urlcolor= blue, 
    linkcolor= blue, 
    citecolor=blue, 
    pdftitle={{J}oint {S}ource-{C}hannel {C}oding Schemes for {A}chieving {S}trong Secrecy at {N}egligible {C}ost}, 
    pdfauthor={R\'emi CHOU, Matthieu Bloch, Joerg Kliewer}, 
    pdfsubject={} 
}

\input{fct_defs}

\input{rv_defs}

\graphicspath{{ }}

\begin{document}

\newtheorem{definition}{Definition}
\newtheorem{theorem}{Theorem}
\newtheorem{proposition}{Proposition}
\newtheorem{lemma}{Lemma}
\newtheorem{remark}{Remark}
\newtheorem*{example}{Example}

\renewcommand{\qedsymbol}{ \begin{tiny}$\blacksquare$ \end{tiny} }

\acrodef{DMS}[DMS]{Discrete Memoryless Source}
\acrodef{DMC}[DMC]{Discrete Memoryless Channel}

\renewcommand{\leq}{\leqslant}
\renewcommand{\geq}{\geqslant}

\allowdisplaybreaks[4]
\title{{C}oding Schemes for {A}chieving {S}trong Secrecy at {N}egligible {C}ost}
\author{R\'{e}mi A. Chou, Badri N. Vellambi,~\IEEEmembership{Senior Member,~IEEE},\\ Matthieu~R.~Bloch,~\IEEEmembership{Senior Member,~IEEE}, and Jörg Kliewer,~\IEEEmembership{Senior Member,~IEEE}\thanks{R. A. Chou is with the Department of Electrical Engineering, The Pennsylvania State University, University Park, PA. M.~R.~Bloch is with the School~of~Electrical~and~Computer~Engineering,~Georgia~Institute~of~Technology, Atlanta,~GA and with GT-CNRS UMI 2958, Metz, France. J. Kliewer and B. N. Vellambi are with the Department of Electrical and Computer Engineering, New Jersey Institute of Technology, Newark, NJ.} \thanks{E-mails : remi.chou@psu.edu; badri.n.vellambi@ieee.org; matthieu.bloch@ece.gatech.edu; jkliewer@njit.edu. Parts of the results were presented at the 2012 IEEE International Symposium on Information Theory~\cite{Bloch12}, the 2013 IEEE International Symposium on Information Theory~\cite{Chou13} and the 2015 IEEE International Symposium on Information Theory~\cite{Vellambi15}. This work was supported in part by NSF under grant CCF-1320298, CCF-1527074, CCF-1439465, CCF-1440014, and by ANR with grant 13-BS03-0008.}}

\maketitle

\begin{abstract}
We study the problem of achieving strong secrecy over wiretap channels at \textit{negligible cost}, in the sense of maintaining the overall communication rate of the same channel without secrecy constraints. Specifically, we propose and analyze two source-channel coding architectures, in which secrecy is achieved by multiplexing public and confidential messages. In both cases, our main contribution is to show that secrecy can be achieved without compromising communication rate and by requiring only randomness of asymptotically vanishing rate. Our first source-channel coding architecture relies on a modified wiretap channel code, in which randomization is performed using the output of a source code. In contrast, our second architecture relies on a standard wiretap code combined with a modified source code termed \textit{uniform compression code}, in which a small shared secret seed is used to enhance the uniformity of the source code output. We carry out a detailed analysis of uniform compression codes and characterize the optimal size of the shared seed. 
\end{abstract}

\section{Introduction}
While cryptography is traditionally implemented at the application layer, physical-layer security aims at ensuring secrecy by taking advantage of the inherent noise at the physical-layer of communication channels. The benefits of physical-layer security are substantiated by numerous theoretical results~\cite{Liang09,Bloch11}, in particular those related to the wiretap channel model~\cite{Wyner75}, which suggest that one can achieve information-theoretic secrecy without sharing secret keys. Although early works on physical-layer security were mostly restricted to eavesdropping attacks under optimistic assumptions regarding channel knowledge, and only established the existence of codes for physical-layer security by means of non-constructive random coding arguments, there has been much progress recently. In particular, attacker models have been extended to situations with limited channel knowledge, e.g., with compound channels~\cite{Liang09c,Bjelakovic2013,Bloch2011e,nafea2015wiretap,goldfeld2015wiretap}, state-dependent channels~\cite{Chia2012,Chen2008}, or arbitrarily varying channels~\cite{raey,He10,MolavianJazi2009}; several explicit low-complexity codes with strong information-theoretic secrecy guarantees have also been designed, for instance, based on low-density parity check codes~\cite{Subramanian2011}, polar codes~\cite{Mahdavifar11,Sasoglu13,Chou14c} or invertible extractors~\cite{Hayashi11,Bellare2012}. 

Despite these recent advances, physical-layer security schemes are yet to be integrated into communication systems. One factor that may have hindered their adoption is the limited attention paid to the \emph{cost} of physical-layer security, assessed in terms of the decrease in achievable communication rates, and the additional resources required for its implementation. In fact, if one hopes to deploy physical-layer systems, it is reasonable to ask that their operation: i) be transparent or at least compatible with upper layer protocols, ii)  not affect communication rates, and iii)  not require additional resources. However, most studies of physical-layer security focus on the characterization of secrecy capacity, which is always less than the capacity, thereby suggesting that secrecy can only be achieved at the cost of reducing communication rates; furthermore, most existing models and coding schemes implicitly assume the presence of an unlimited source of uniform random numbers to realize a stochastic encoder.

The objective of this paper is to revisit these assumptions and to show that the cost of secrecy can be made negligible, i.e., secrecy neither incurs a reduction in overall communication rate nor requires extra randomness resources. The crux of our approach is to analyze the wiretap channel model illustrated in Fig.~\ref{figmodel0}, in which the encoder only uses a random seed of vanishing rate. More specifically, the objective is to \emph{multiplex} a confidential source with a public source, while maximizing the sum-rate of secret and public communication. The idea of multiplexing messages to achieve secrecy already implicitly appears in the original work of Csisz\'ar and K\"orner~\cite{Csiszar78}, and is explicitly formalized in~\cite{Kobayashi05,Kobayashi2013,Xu2008}; however, our approach differs in that: (i) we relax the common assumption that messages are exactly uniformly distributed, which is unrealistic even if messages are compressed with optimal source codes~\cite{Han05,Hayashi08}; and (ii) we consider a strong notion of secrecy. 
\begin{figure} 
\centering
    \scalebox{0.61}{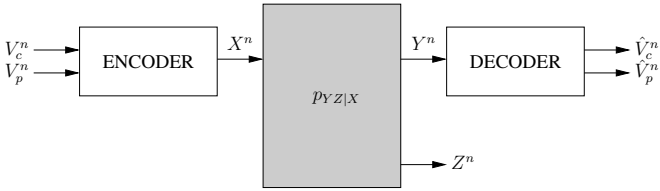}  
  \caption{Multiplexing of confidential and public sources. The confidential source, $V_c^n$, must be reconstructible by the receiver and must be kept secret from the eavesdropper. The public source, $V_p^n$, should be reconstructible by the receiver, and information may be leaked to the eavesdropper.}
  \label{figmodel0}
  \vspace*{-0em}
\end{figure}

The main contributions of this paper are two source-channel coding architectures that  achieve informa\-tion-theoretic secrecy over this channel model. The first one, illustrated in Fig.~\ref{fig:arch1}, is based on wiretap codes and requires a random seed of negligible rate to compress the public source. The second one, illustrated in Fig.~\ref{fig:arch2}, combines a wiretap code designed to operate with uniform randomization with a modified source encoder, which compresses data while simultaneously ensuring near-uniform outputs. This second architecture is slightly more restrictive than the first simply because it  requires the encoder and the decoder to share in advance a small secret seed. For both architectures the presence of a random seed at the encoder is meant to obtain a nearly uniform source from the public source, and is thus unnecessary if the public source is uniform. Nevertheless, regardless of the architecture, a secret key for authentication is required~\cite{Gemmell94,Gilbert74}. While both architectures achieve the same optimal performance, the former modifies the physical layer of the protocol stack whereas the latter modifies the application layer, which makes it much easier to implement protocol changes. We also highlight that the concept of uniform compression introduced and studied in Section \ref{Sec_Sol2} is of independent interest, as it can be used in other security problems. For instance, in secure network coding~\cite{cai2011secure,feldman2004capacity,silva2011universal,cui2013secure}, security is typically obtained by injecting uniformly distributed ``packets" into the network, which the destination nodes are able to decode along with the messages. Similar to the compression of the public source with uniform compression codes in Section \ref{Sec_Sol2}, these uniformly distributed ``packets" in secure network coding could be replaced by uniformly compressed public messages. 
\begin{figure}
  \centering
  \subfigure[Architecture based on a modified wiretap code.]{\label{fig:arch1}\scalebox{0.435}{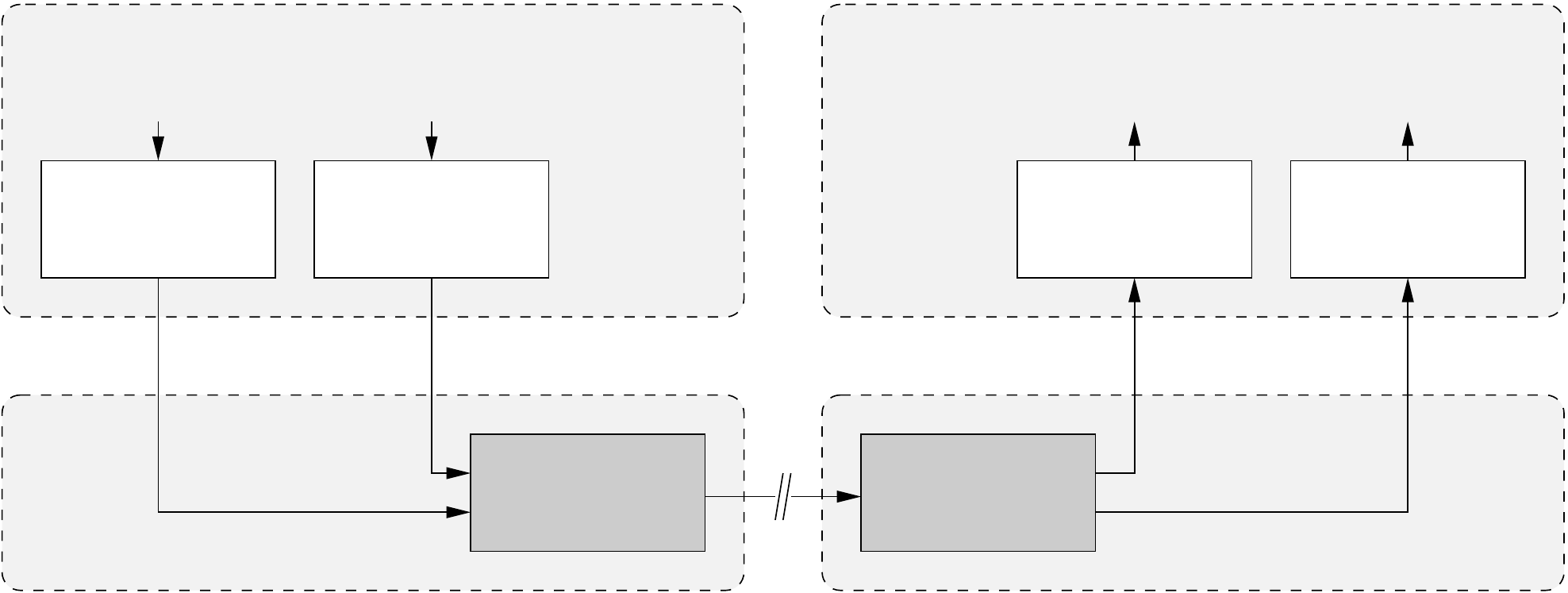}}
  \subfigure[Architecture based on a modified source encoder.]{\label{fig:arch2}\scalebox{0.435}{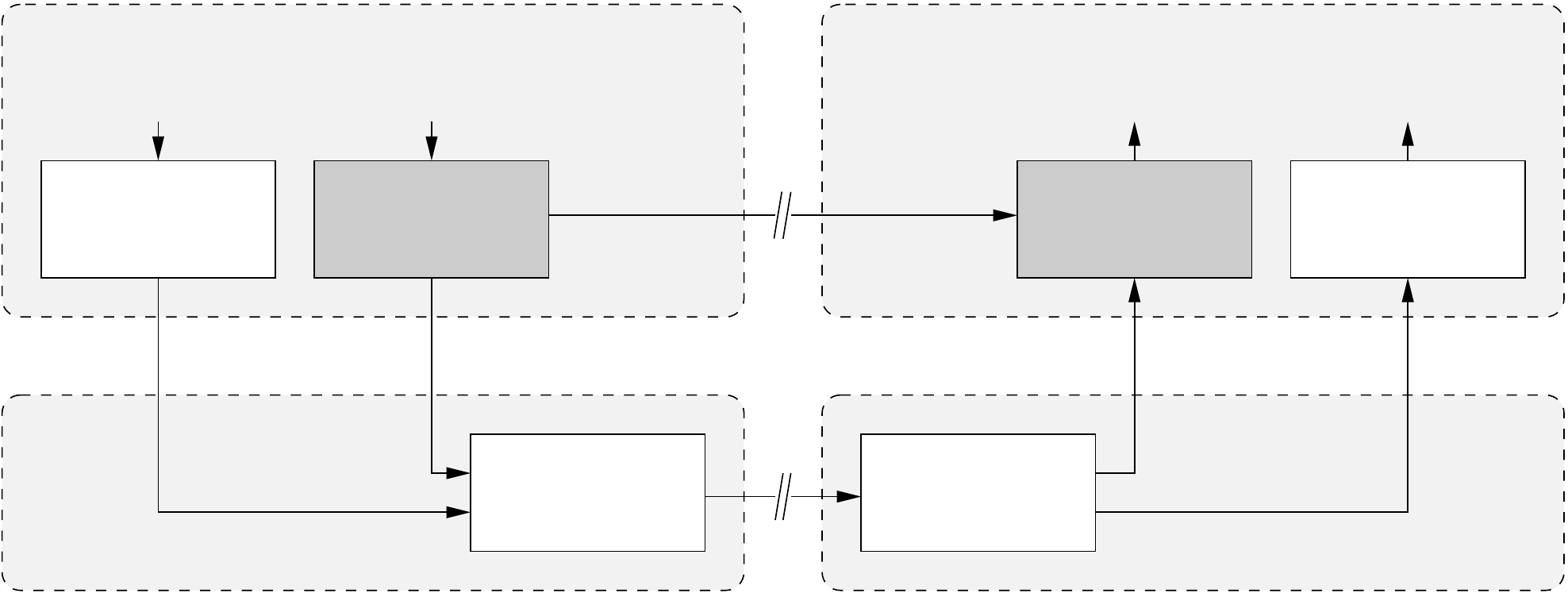}}
  \caption{Proposed architectures to multiplex confidential source sequences $V_c^n$ and public source sequences $V_p^n$. $U_d$ is a uniformly distributed seed, whose length $d$ is sub-linear in the code length $n$.}
  \label{fig:architectures}
\end{figure}

Our model initially presented in \cite{Bloch12} is closely related to the concurrent study \cite{Watanabe12a} and the subsequent study \cite{Hayashi2012}, with journal versions \cite{Watanabe12}, \cite{Hayashi12}. However, our model is not subsumed by any of the models considered in \cite{Watanabe12a,Watanabe12} or in~\cite{Hayashi2012,Hayashi12}. 
The main difference with \cite{Hayashi2012,Hayashi12} is that we only allow a vanishing rate of randomness to be used at the encoder to account for all the resources required to achieve strong secrecy. This assumption results in an \emph{additional constraint on the rate of the public source}, which is not accounted for by the analysis of \cite{Hayashi2012,Hayashi12}. We provide additional details on how our achievability schemes differ from \cite{Hayashi2012,Hayashi12} in Remark 1. 
Our model also differs from \cite{Watanabe12a,Watanabe12}, as we consider \emph{non-uniform sources} instead of uniform messages, so the analysis in \cite{Watanabe12a,Watanabe12} does not apply. We further detail in Remark 2 how our achievability schemes differ from those in \cite{Watanabe12a,Watanabe12}. 
Because of differences in the models considered, the two achievability arguments we present are conceptually different from those in \cite{Watanabe12a,Watanabe12,Hayashi2012,Hayashi12}, and shed a different light on how to implement multiplexing.
%
\begin{remark}
In \cite{Hayashi2012,Hayashi12}, the authors analyze the transmission over a wiretap channel of a common message $S_0$ and multiple confidential messages $S_1,\ldots, S_T$ that may not be jointly independent. Moreover, the encoder is allowed to encode these $T+1$ messages using a {\textbf{randomized}} encoder. 
In our approach, we have two independent sources, which when compressed losslessly but separately, yield two separate non-uniform messages. One of these sources is confidential, while the other is public and can possibly be leaked to the eavesdropper. However, in the two architectures considered in our work, the encoding is only allowed to use a random seed whose length grows sub-linearly in the code length $n$. This introduces a new constraint on the minimum rate of the public source that is {\textbf{absent}} in  \cite{Hayashi2012,Hayashi12}.
Furthermore, the randomized encoding in \cite{Hayashi2012,Hayashi12} uses a commutative group structure, while our two achievability schemes use either (i) typicality-based compression arguments to show that the R\'{e}nyi entropy of order~2 of the compressed public source approaches its entropy (see Section \ref{Sec_Sol1}); or (ii) lossless compression codes with near uniform encoder output that require a random seed whose length grows as $O(\sqrt{n})$, where $n$ is the code length (see Section \ref{Sec_Sol2}).
\end{remark}
\begin{remark}
In \cite{Watanabe12a,Watanabe12}, the authors study the broadcast channel
with confidential messages and precisely analyze the trade-offs
among the rates of uniform secret messages, uniform public messages, and  uniform local randomness. In contrast, we study a \textbf{source setting} in which  \textbf{non-uniform}  confidential and public sources are transmitted over a wiretap channel. We present two distinct achievability arguments in Section~\ref{Sec_Sol1} and Section \ref{Sec_Sol2} for the proposed generalization that do not naturally follow from the proof arguments in  \cite{Watanabe12a,Watanabe12}. Despite similarities with the converse for our model, the converse in~\cite{Watanabe12a,Watanabe12} does not directly apply to the setting considered in Section \ref{Sec_Sol2} because of the presence of a shared seed. Therefore, for completeness a converse for our model is provided in Appendix \ref{App_converse}.
%
%
\end{remark}

 
The remainder of the paper is organized as follows. In Section~\ref{Sec_PS}, we formally describe the communication model under consideration. In Sections~\ref{Sec_Sol1} and \ref{Sec_Sol2}, we prove that the two architectures shown in Fig.~\ref{fig:architectures} achieve near-optimal performance, i.e., offer the same rate trade-offs as the communication problem without security constraints. More specifically, we show in Section~\ref{Sec_Sol1} the existence of wiretap codes that ensure secrecy with non-uniform randomization, while in Section \ref{Sec_Sol2}, we show how to render the output of a source code nearly uniform. Section~\ref{sec:conclusion} concludes the paper with some perspectives for future work.

\section{Preliminaries and Problem Statement} \label{Sec_PS}

\subsection{Notation}
\label{sec:notation}

Random variables, e.g., $X$, and their realizations, e.g., $x$, are denoted by uppercase and lowercase serif font, respectively, while alphabets, e.g., $\calX$, are denoted by calligraphic font. Unless otherwise specified, random variables have finite alphabets, and the generic probability mass function of $X$ is denoted by $p_X$. Basic information-theoretic quantities, e.g., $H(X)$, $I(X;Y)$ are defined as in~\cite{kramerbook}. For two random variables $X$ and $X'$ over the alphabet $\calX$, the variational distance between $X$ and $X'$ is
$  \V{p_X,p_{X'}}\eqdef \sum_{x\in\calX}\abs{p_X{(x)}-p_{X'}(x)}.
$ For any $\epsilon>0$, $\delta(\epsilon)$ denotes a positive function of $\epsilon$ such that $\lim_{\epsilon\downarrow 0}\delta(\epsilon)=0$. We also define $\llbracket a,b \rrbracket \triangleq [ \lfloor a \rfloor , \lceil b \rceil] \cap \mathbb{N}$.

\subsection{Wiretap channel model}
\label{sec:channel-model}
Let $\mathcal{X}$, $\mathcal{Y}$ and $\mathcal{Z}$ be finite alphabets. As illustrated in Fig.~\ref{figmodel0}, we consider a \ac{DMC} $\left(\mathcal{X},p_{YZ|X},\mathcal{Y}\times\mathcal{Z}\right)$. The channel $\left(\mathcal{X},p_{Y|X},\mathcal{Y}\right)$ is the main channel while the channel $\left(\mathcal{X},p_{Z|X},\mathcal{Z}\right)$ is the eavesdropper's channel. We assume that the transmitter Alice wishes to transmit the realizations of two independent \acp{DMS} $\left(\smash{V_c,p_{V_c}}\right)$ and $\left(\smash{V_p,p_{V_p}}\right)$. Both sources are to be reconstructed without errors by the receiver Bob observing $Y^n$, while the source $\left(\smash{V_c,p_{V_c}}\right)$ should be kept secret from the eavesdropper Eve observing $Z^n$. Hence, we refer to $\left(\smash{V_c,p_{V_c}}\right)$ as the \emph{confidential source} and to $\left(\smash{V_p,p_{V_p}}\right)$ as the \emph{public source}.

\begin{definition}
  A code for $\mathcal{C}_n$ the wiretap channel consists of the following.
  \begin{itemize}
  \item A {deterministic} encoding function $f_n:\mathcal{V}_c^n\times\mathcal{V}_p^n \times \intseq{1}{2^{d_n}} \rightarrow\mathcal{X}^n$, which maps $n$ symbols of the confidential source and $n$ symbols of the public source to a codeword of length $n$ with the help of a uniformly distributed seed of length $d_n$ bits;
  \item A decoding function $g_n:\mathcal{Y}^n\rightarrow \mathcal{V}_c^n\times\mathcal{V}_p^n$, which maps a sequence of $n$ channel output observations to $n$ symbols of the confidential source and $n$ symbols of the public source.
  \end{itemize}
\end{definition}
The performance of $\mathcal{C}_n$ is measured in terms of the average probability of error
  \begin{align*}
    \mathbf{P}_e ( \mathcal{C}_n)\triangleq \mathbb{P} \left[ (V_c^n,V_p^n) \neq g_n(Y^n)\right],
  \end{align*}
and in terms of the secrecy metric
  \begin{align*}
    \mathbf{S}(\mathcal{C}_n)\triangleq \max_{v_c^n\in\mathcal{V}_c^n} \V{p_{Z^n|V_c^n=v_c^n},p_{Z^n}}.
  \end{align*}
Note that since we do not know the exact output distribution of the source encoders, we impose a security constraint akin to semantic security~\cite{Bellare2012}.
We also require the length of the uniformly distributed seed to be sub-linear in $n$, i.e.,
$$
\lim_{n \to \infty} \frac{d_n}{n }=0.
$$

Note that in our second architecture presented in Section \ref{Sec_Sol2} and depicted in Figure \ref{fig:arch2}, we allow the seed to be shared between the encoder and the decoder, in which case, the seed is also an argument to the decoding function $g_n$.

\subsection{Source-channel coding theorem}
\label{sec:source-chann-coding}

\begin{theorem}
\label{th:joint-source-channel}
  Consider a confidential \ac{DMS} $\left(\smash{V_c,p_{V_c}}\right)$ and a public \ac{DMS} $\left(\smash{V_p,p_{V_p}}\right)$ to be transmitted over a wiretap channel $\left(\mathcal{X},p_{YZ|X},\mathcal{Y}\times\mathcal{Z}\right)$. For any random variable $Q$ over a finite alphabet $\mathcal{{Q}}$ such that ${Q}-X-YZ$, if 
  \begin{align*}
    \left\{
      \begin{array}{l}
        H(V_c)+H(V_p)<I(X;Y|Q)\\
        H(V_c) < I(X;Y|Q)-I(X;Z|Q)\\
        H(V_p)>I(X;Z|Q)
      \end{array}
\right.,
  \end{align*}
  then there exists a sequence of codes $\{\mathcal{C}_n\}_{n\geq 1}$ such that 
  \begin{align} \label{eqlim}
\lim_{n\rightarrow\infty}\mathbf{P}_e(\mathcal{C}_n)=\lim_{n\rightarrow\infty}\mathbf{S}(\mathcal{C}_n)=0.\end{align}
 Conversely, if there exists a sequence of codes $\{\mathcal{C}_n\}_{n\geq 1 }$ such that \eqref{eqlim} holds, then there must exist a random variable $Q$ over $\mathcal{Q}$ with $|\mathcal{Q}| \leq 3$ such that ${Q}-X-YZ$ and
\begin{align*}
      \left\{
      \begin{array}{l}
        H(V_c)+H(V_p)\leq I(X;Y|{Q})\\
        H(V_c) \leq  I(X;Y|{Q})-I(X;Z|{Q})\\
        H(V_p)\geq I(X;Z|{Q})
      \end{array}
\right..
\end{align*}
\end{theorem}
Although the result might seem intuitive, the achievability proof does not follow from standard arguments and known results because of the use of vanishing-rate randomness at the encoder. The main contributions of this paper are the two achievability proofs detailed next, the first one in Section~\ref{Sec_Sol1} using the architecture of Fig.~\ref{fig:arch1}, the second one in Section~\ref{Sec_Sol2} using the architecture of Fig.~\ref{fig:arch2}. Note that the converse in \cite{Watanabe12} does not directly apply to the setting of Section \ref{Sec_Sol2}, because of the presence of a pre-shared seed. We provide a detailed proof for the converse of Theorem \ref{th:joint-source-channel} in Appendix \ref{App_converse}. 

\begin{remark}
  Unlike the capacity region of the broadcast channel with confidential messages, the information constraints in Theorem~\ref{th:joint-source-channel} do not include an auxiliary random variable $V$ such that ${Q}-V-X-YZ$. This result is not surprising, as this extra random variable accounts for the addition of artificial noise (channel prefixing) in the encoder, which is not allowed by our model, as we require all encoder inputs to be decoded at the receiver. The random variable ${Q}$ is merely a time-sharing random variable~\cite{Bloch12,Watanabe12}. Similar to \cite[Appendix~C]{Elgamal11}, it is sufficient to consider an alphabet $\mathcal{Q}$ such that $|\mathcal{Q}| \leq 3$ by Fenchel–Eggleston–Carath\'{e}odory theorem.
\end{remark}

\section{Coding architecture based on wiretap codes with non-uniform randomization} \label{Sec_Sol1}
If one were to rely on known wiretap codes~\cite{Wyner75,Csiszar78} to transmit the confidential and public sources, and meet the strong secrecy constraint for the confidential source, one would have to ensure that the randomization of the encoder could be performed with a nearly uniform source of random numbers, measured at least in terms of total variation. If no reconstruction constraints were placed on the public source $\left(\smash{V_p,p_{V_p}}\right)$, a natural approach would simply be to extract the intrinsic randomness of the source~\cite{Vembu1995} to generate nearly uniform random numbers; this strategy happens to be optimal as shown in~\cite[Proposition 1]{Bloch12}. However, unlike the model in~\cite{Bloch12}, the present setting requires the reconstruction of the public source at the receiver. Although lossless compression of the public source might intuitively seem to solve the problem, it would actually not lead to a uniform random number. As alluded to earlier,~\cite{Han05} shows that lossless compression of a source at the optimal rate does not necessarily ensure uniformity under variational distance. In addition, for \acp{DMS},~\cite[Theorem~4]{Hayashi08} shows that there exists a fundamental trade-off between reconstruction error probability and uniformity of the encoder output measured in variational distance. 
To circumvent this limitation, we design wiretap codes that operate with a non-uniform randomization. 
\subsection{Wiretap codes with non-uniform randomization}
We start by studying the wiretap channel model illustrated in Fig.~\ref{fig:non-uniform-randomization}, in which the objective is to encode a secret message $M_c\in \llbracket 1, 2^{nR_c} \rrbracket$ by means of a public message $M_p \in\llbracket 1, 2^{nR_p} \rrbracket$; we do not assume that messages are uniform, but we assume that the statistics of the public message $M_p$ are known to the encoder. We call the corresponding wiretap  code a $(2^{nR_c},2^{nR_p},n)$ wiretap code. In this case, we show that secrecy is still achievable, but at a rate $\frac{1}{n} H_2(M_p)$, where $H_2(M_p)$ denotes the  R\'enyi entropy of order 2 and is given by
\begin{align*}
  H_2(M_p) \triangleq - \log\left[\sum_{m\in \llbracket 1, 2^{nR_p} \rrbracket }p_{M_p}(m)^2\right].
\end{align*}
\begin{figure}
  \centering
  \scalebox{0.59}{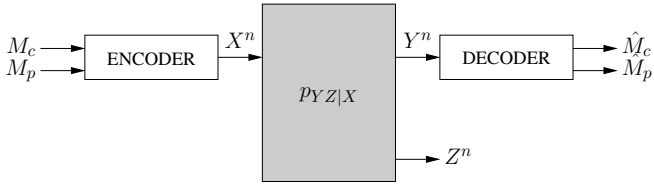}
  \caption{Wiretap channel model with non-uniform randomization.}
  \label{fig:non-uniform-randomization}
\end{figure}

\begin{proposition}  \label{prop_non_uniform2}
Let $p_{{Q}XYZ}$ be a joint distribution that factorizes as $p_{{Q}}p_{X|{Q}}p_{YZ|X}$. Then, if 
  \begin{align*}
    R_c+R_p&<I(X;Y|{Q}),\\
    R_c&<I(X;Y|{Q})-I(X;Z|{Q}),\\
    I(X;Z|{Q})&<  \lim_{n \rightarrow \infty} \frac{1}{n} H_2(M_p),
  \end{align*}
  there exists a sequence of wiretap codes $\{\mathcal{C}_n\}_{n\geq 1}$ such that
  \begin{align*}
    \lim_{n\rightarrow\infty}&\max_{m}\mathbb{P}\left[\hat{M}_c\neq M_c|M_c=m\right]=0,\\
    \lim_{n\rightarrow\infty}&\max_m\mathbb{P}\left[\hat{M}_p\neq M_p|M_c=m\right]=0,\\
    \lim_{n\rightarrow\infty}&\max_{m}\V{p_{Z^n|M_c=m},p_{Z^n}}=0.
  \end{align*}
\end{proposition}
\begin{proof}
  See Appendix~\ref{App_renyi_constraint}.
\end{proof}
As shown in~\cite[Proposition 1]{Bloch12}, if one did not require the reconstruction of $M_p$, one could achieve secret rates $R_c$ as in Proposition~\ref{prop_non_uniform2}, but with the constraint $I(X;Z|{Q})<  \lim_{n \rightarrow \infty} \frac{1}{n} H(M_p)$ instead. In general, $\frac{1}{n}H_2(M_p)\leq \frac{1}{n} H(M_p)$, and the penalty paid by using the R\'enyi entropy instead of the Shannon entropy may be significant. The following example highlights an extreme example of such a situation.
\begin{example}
  Consider $M_p\in \llbracket 1, 2^{nR_p} \rrbracket$ such that
  \begin{align*}
    \mathbb{P}[M_p=1]\triangleq 2^{-n\alpha R_p}, \text{ }\mathbb{P}[M_p=i] \triangleq \frac{1-2^{-n\alpha R_p}}{2^{nR_p}-1}\text{ if }i\neq 1,
  \end{align*}
  where $\alpha\in]0,\frac{1}{2}[$ is a parameter that controls the uniformity of the distribution. Note that
  \begin{align*}
    \lim_{n\rightarrow\infty}\tfrac{1}{n}H_2(M_p) = \alpha R_p \text{ whereas } \lim_{n\rightarrow\infty}\tfrac{1}{n}H(M_p) = R_p.
  \end{align*}
\end{example}
Consequently, the achievable rates predicted in Proposition~\ref{prop_non_uniform2} could be arbitrarily smaller than those in~\cite[Proposition 1]{Bloch12}. Fortunately, a combination of a source code with a wiretap code identified in Proposition~\ref{prop_non_uniform2} is sufficient to achieve the optimal rate of Theorem~\ref{th:joint-source-channel}. 

\subsection{Achievability of Theorem \ref{th:joint-source-channel} based on wiretap codes with non-uniform randomization}

We first refine a known result regarding the existence of good source codes.
\begin{lemma}
\label{lm:compression}
  Consider a \ac{DMS} $\left(\mathcal{V},p_V\right)$. Then, there exists a sequence of source encoders $f_n:\mathcal{V}^n\times \intseq{1}{2^{d_n}}\rightarrow\intseq{1}{2^{nR_n}}$ and associated decoders $g_n$ such that
  \begin{align*}
    \lim_{n\rightarrow\infty}R_n = H(V),  
         \lim_{n\rightarrow\infty}\frac{1}{n}H_2(f_n(V^n,U_{d_n})) = H(V),\\ 
          \lim_{n\rightarrow\infty} \P{V^n\neq g_n(f_n(V^n,U_{d_n}))}=0, \lim_{n \to \infty} \frac{d_n}{n } =0.
  \end{align*}
\end{lemma}
\begin{proof}
   We consider a typical-sequence-based source. Specifically, let $n\in\mathbb{N}$, let $\epsilon_0 >0$ function of $n$ to be determined later, and let $ \mathcal{T}_{\epsilon_0}^n(V)$ be the set of $\epsilon_0$-letter-typical sequences of length $n$ with respect to $p_{V}$~\cite{kramerbook}. The typical sequences are labeled $v^n(m)$ with $m\in\intseq{1}{2^{nR_n}}$ and $R_n\triangleq \frac{1}{n}\log |\mathcal{T}_{\epsilon_0}^n(V)| $. The encoder $f_n$ outputs $m$ if the input sequence $v^n=v^n(m)\in\mathcal{T}_{\epsilon_0}^n(V)$, otherwise it generates $m\in\intseq{1}{2^{nR_n}}$ uniformly a random. Note that this uniform selection when the realization of $V^n$ is atypical can be done by a random seed $U_{d_n}$ of appropriate size $d_n$. Decoding is performed by returning the typical sequence $v^n(m)$ corresponding to the received message $m$. By~\cite[Theorem 1.1]{kramerbook}, we know that $\P{V^n\neq g_n(f_n(V^n, U_{d_n}))}\leq \delta_{\epsilon_0}(n)$ with $\delta_{\epsilon_0}(n) \triangleq 2 |\mathcal{V}| e^{-n\epsilon_0^2 \mu_{V}}$, $\mu_{V} \triangleq \displaystyle\min_{r \in \text{supp}(p_{V})}p_{V}(r)$, where supp denotes the support of a distribution, and $R_n<(1+\epsilon_0) H(V)$. Hence, for any $m\in \llbracket 1, 2^{nR_n} \rrbracket$
\begin{align*}
&p_{f_n(V^n, U_{d_n})}(m) \\
& = \mathbb{P}\left[\text{$V^n=v^n(m)$ or } (V^n  \notin \mathcal{T}_{\epsilon_0}^n(V) \right. \\& \phantom{---}\left. \text{ and $m$ is drawn uniformly from $\intseq{1}{2^{nR_n}}$}) \right] \\
& \leq 2^{-n(1-\epsilon_0){H}(V)} + \frac{\delta_{\epsilon_0}(n)}{|\mathcal{T}_{\epsilon_0}^n(V)|}\\
& \leq 2^{-n(1-\epsilon_0){H}(V)} + \frac{\delta_{\epsilon_0}(n)}{ 1- \delta_{\epsilon_0}(n)} 2^{-n(1-\epsilon_0) {H}(V)} \\
& = \frac{2^{-n(1-\epsilon_0)H(V)}}{1-\delta_{\epsilon_0}(n)}.
\end{align*}
Hence, since for any discrete random variable $X$ over $\mathcal{X}$, $H(X) \geq H_2(X) \geq H_{\infty}(X)$, we have
\begin{align*}
n H(V)
&\geq H(f_n(V^n))\\
&\geq H_2(f_n(V^n))\\
 &\geq {H}_{\infty} (f_n(V^n)) \\
&=  - \log (\max_m p_{f_n(V^n)} (m) ) \\
&\geq  n(1-\epsilon_0){H}(V) + \log (1-\delta_{\epsilon_0}(n)).
\end{align*}
We may choose $\epsilon_0 = n^{- 1/2 + \epsilon_b}$ and $\epsilon_b >0$. 

Note that the encoder requires $U_{d_n}$ to encode the non-typical sequences. To mitigate this requirement, we apply the encoder to $b(n)$ sequences of length $a(n)$, where $a(n)$ and $b(n)$ are any integers such that $a(n)b(n)=n$ and $\lim_{n\to \infty} a(n) = + \infty = \lim_{n\to \infty} b(n) $.\footnote{A possible choice is $a(n) \triangleq n^{1- \lambda}$ and $b(n) \triangleq n^{\lambda}$ with $\lambda \in ]0,1[$.} Hence, the amount of required randomness is negligible compared to $n$ since $\mathbb{P} [V^{a(n)}  \notin \mathcal{T}_{\epsilon_0}^{a(n)}(V)] \leq \delta_{\epsilon_0}(a(n))$.
\end{proof}

In the remainder of the paper, we refer to the source codes identified in Lemma~\ref{lm:compression} as ``typical-sequence based'' source codes. 

Going back to the setting of Section~\ref{sec:channel-model}, let us apply Lemma~\ref{lm:compression} to both sources $(\mathcal{V}_c,p_{V_c})$ and $(\mathcal{V}_p,p_{V_p})$. Let $\epsilon>0$. There exists $N_1 \in \mathbb{N}$ and two source encoder-decoder pairs, denoted $(f_n^c,g_n^c)$ and $(f_n^p,g_n^p)$, respectively, such that for $n>N_1$, $\P{(V_c^n,V_p^n)\neq (g_n^c(f_n^c(V_c^n)),g_n^p(f_n^p(V_p^n,, U_{d_n})))}\leq \epsilon$. We set $M_c\triangleq f_n^c(V_c^n)\in\intseq{1}{2^{nR_c}}$ and $M_p\triangleq f_n^p(V_p^n, U_{d_n})\in\intseq{1}{2^{nR_p}}$. Note that we only need randomness for the public source. If there exists a distribution $p_{QXYZ}$ that satisfies the condition of Proposition~\ref{prop_non_uniform2}, then, there exists $N_2 \in \mathbb{N}$ and a wiretap code with encoder-decoder pair $(f_n,g_n)$ such that for $n > N_2$,
\begin{align*}
  \max_{m}\mathbb{P}\left[\hat{M}_c\neq M_c|M_c=m\right]&<\epsilon,\\
\max_m\mathbb{P}\left[\hat{M}_p\neq M_p|M_c=m\right]&<\epsilon,\\
\max_{m}\V{p_{Z^n|M_c=m},p_{Z^n}}&<\epsilon.
\end{align*}
Encoding the sources into codewords as $f_n(f_n^c(V_c^n),f_n^p(V_p^n, U_{d_n}))$, and forming estimates from the channel output $Y^n$ as $\hat{V}_c^n=g_n^c(g_n(Y^n))$, and $\hat{V}_p^n=g_n^p(g_n(Y^n))$, we observe that for $n> \max(N_1,N_2)$, $\P{(V_c^n,V_p^n)\neq (\hat{V}_c^n,\hat{V}_p^n)}\leq 3\epsilon$
and for any $v_c^n\in\mathcal{V}_c^n$, $\V{p_{Z^n|V_c^n=v_c^n},p_{Z^n}}\leq \epsilon$. By taking the limit $\epsilon \rightarrow 0$, we conclude with Lemma~\ref{lm:compression} that a code for the wiretap channel can be constructed provided
\begin{align*}
  H(V_c)+H(V_p)&<I(X;Y|{Q}),\\ H(V_c) &< I(X;Y|{Q})-I(X;Z|{Q}),\\ H(V_p) &> I(X;Z|{Q}).
\end{align*}

\section{Coding architecture based on uniform compression codes} \label{Sec_Sol2}
In this section, we develop a second optimal architecture. As before, our objective is to circumvent the impossibility of generating uniform random numbers with source codes~\cite[Theorem 4]{Hayashi08}, but this time by modifying the operation of the source codes themselves. The approach to overcome this impossibility is to introduce a small shared uniformly distributed random seed. The benefit of this second architecture is that it only requires a modification at the application layer of the protocol stack. However, the price paid is that the transmitter and the receiver must now share a seed whose rate can be shown to be made vanishingly small. This contrasts with our first architecture in Section \ref{Sec_Sol1} for which the seed is \emph{not} available at the decoder. 

\subsection{Uniform compression codes} 
\label{Sec_uc}

\begin{figure}[]
\centering
  \includegraphics[width=8.8cm]{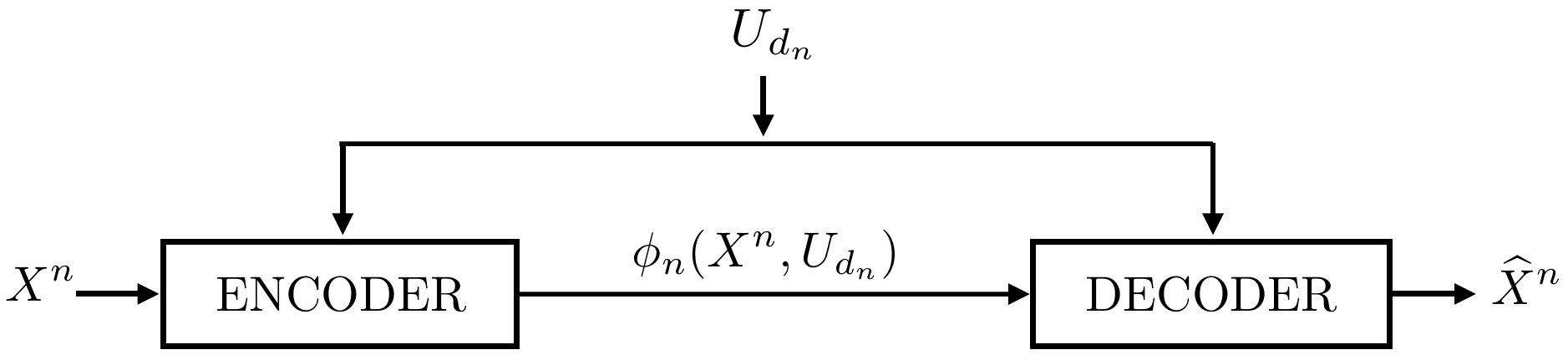}
  \caption{Source encoder and decoder with uniform outputs.}
  \label{figmodel}
  \vspace*{-0em}
\end{figure}
Consider a \ac{DMS} $(\mathcal{X},p_{X})$. Let $n \in \mathbb{N}$, $d_n \in \mathbb{N}$, and let $U_{d_n}$ be a uniform random variable over $\mathcal{U}_{d_n} \triangleq \llbracket 1, 2^{d_n}\rrbracket$ independent of $X^n$. In the following we refer to $U_{d_n}$ as the \emph{seed} and $d_n$ as its length. As illustrated in Figure~\ref{figmodel}, our objective is to design a source code to compress and reconstruct the \ac{DMS} $(\mathcal{X},p_{X})$ with the assistance of a seed $U_{d_n}$.
\begin{definition} \label{def_ucc}
A $(2^{nR},n,2^{d_n})$ uniform compression code $\mathcal{C}_n$ for a \ac{DMS} $(\mathcal{X},p_{X})$ consists of
\begin{itemize}
\item A message set $\mathcal{M}_n \triangleq \llbracket 1, M_n  \rrbracket$, with $M_n \triangleq 2^{nR}$,
\item A seed set $\mathcal{U}_{d_n} \triangleq \llbracket 1, 2^{d_n} \rrbracket$,
\item An encoding function $\phi_n : \mathcal{X}^n \times \mathcal{U}_{d_n} \to \mathcal{M}_n$,
\item A decoding function $\psi_n : \mathcal{M}_n \times \mathcal{U}_{d_n} \to \mathcal{X}^n$.
\end{itemize}
\end{definition}
The performance of the code is measured in terms of the average probability of error and the uniformity of its output as
\begin{eqnarray*}
&\mathbf{P}_e(\phi_n, \psi_n) & \triangleq \mathbb{P}[  X^n \neq \psi_n (\phi_n(X^n,U_{d_n}),U_{d_n})], \\
&\mathbf{U}_e(\phi_n) & \triangleq \V{p_{\phi_n(X^n,U_{d_n})} ,p_{U_{M_n}}},
\end{eqnarray*}
where $U_{M_n}$ has uniform distribution over $\mathcal{M}_n$. 

\begin{remark}
Uniformity could be measured with the stronger metric $\mathbf{U}_e'(\phi_n) \triangleq \mathbb{D}( p_{\phi_n(X^n,U_{d_n})} || p_{U_{M_n}})$, where $\mathbb{D}(\cdot||\cdot)$ is the Kullback-Leibler divergence; however, by~\cite[Lemma 2.7]{bookCsizar}, $\mathbf{U}_e(\phi_n)$ can be replaced by $\mathbf{U}_e'(\phi_n) $, if  $\displaystyle\lim_{n \to \infty} n \mathbf{U}_e(\phi_n)= 0$, which will be the case.
\end{remark}
\begin{definition}	
A rate $R$ is achievable, if there exists a sequence of $(2^{nR},n,2^{d_n})$ uniform compression codes $\{ \mathcal{C}_n \}_{n \geq 1}$ for the \ac{DMS} $(\mathcal{X},p_{X})$,  such that
\begin{align*}
\lim_{n \to \infty} \frac{1}{n} \log M_n  &\leq R, \text{ } 
\lim_{n \to \infty} \frac{d_n}{n}  = 0, \\
\lim_{n \to \infty} \mathbf{P}_e(\phi_n, \psi_n)  &=0,  \text{ } 
\lim_{n \to \infty} \mathbf{U}_e(\phi_n)  = 0.
\end{align*}
\end{definition}

Our main result in this section is the characterization of the infimum of achievable rates with uniform compression codes as well as the optimal scaling of the seed length $d_n$. In the following, we use the Landau notation to characterize the limiting behavior of the seed scaling.
\begin{proposition} \label{Thlossya}
Let $(\mathcal{X} , p_{X})$ be a \ac{DMS}. The infimum of achievable rates with uniform compression codes is $H(X)$.
This infimum is achievable with a seed length $d_n  =   O(\upsilon_n\sqrt{n})$,
for any $\{\upsilon_n\}_{n\in \mathbb{N}} \text { s.t. } \lim_{n\to \infty} \upsilon_n = + \infty$. Moreover, a necessary condition on $d_n$ for a $(2^{nR},n,2^{d_n})$ uniform compression code to achieve $H(X)$~is
$d_n  =  \Omega (\sqrt{n})$, i.e., $\lim_{n\rightarrow \infty} \frac{d_n}{\sqrt n} = +\infty$.
\end{proposition}
\begin{proof}
  See Appendix~\ref{App_seed}.
\end{proof}

\subsection{Explicit uniform compression codes}
\label{sec:expl-unform-compr}

As a first attempt to develop a practical scheme for uniform compression codes, we propose an achievability scheme for Proposition \ref{Thlossya} based on invertible extractors~\cite{Dodis05}. We start by recalling known facts about extractors.

\begin{definition}[\cite{Dodis05}]
Let $\epsilon > 0$. Let $m,d,l \in \mathbb{N}$ and let $t \in \mathbb{R}^+$. A polynomial time probabilistic function $\textup{Ext}: \{0,1\}^m \times \{0,1\}^d \mapsto \{0,1\}^l$ is called a $(m,d,l,t,\epsilon)$-extractor, if for any binary source $X$ satisfying ${H}_{\infty}(X)\geq t$, we have $$\mathbb{V}(p_{\textup{Ext}(X,U_d)},p_{U_l}) \leq \epsilon,$$ where $U_d$ is a sequence of $d$ uniformly distributed bits, $p_{U_l}$ is the uniform distribution over $\{0,1\}^l$. Moreover, a $(m,d,l,t,\epsilon)$-extractor is said to be invertible if the input can be reconstructed from the output and $U_d$. 
\end{definition}
It can be shown~\cite{Dodis05,Dodis04} that there exist \emph{explicit} invertible $(m,d,m,t,\epsilon)$-extractors such that 
\begin{equation} \label{d_exp}
d = m-t +2\log m + 2 \log \frac{1}{\epsilon} + O(1) .
\end{equation}
The following proposition shows that one can establish optimal uniform compression codes using such invertible extractors.

\begin{proposition} \label{Extlossy}
Let $(\mathcal{X}, p_{X})$ be a binary memoryless source. For any $R>H(X)$ and for any $\epsilon > 0$, the rate $R$ can be achieved with a sequence of uniform compression codes such that
\begin{itemize}
\item the seed length scales as $d_n=\Theta(n^{1/2 + \epsilon})$;
\item the encoder $\phi_n: \mathcal{X}^n \times \mathcal{U}_{d_n} \to \mathcal{M}_n$ is composed of a typical-sequence based source code combined with an invertible extractor as described in Figure \ref{fig:inv_ext_scheme}.
\end{itemize}

\begin{figure}
\centering
  \includegraphics[width=8.8cm]{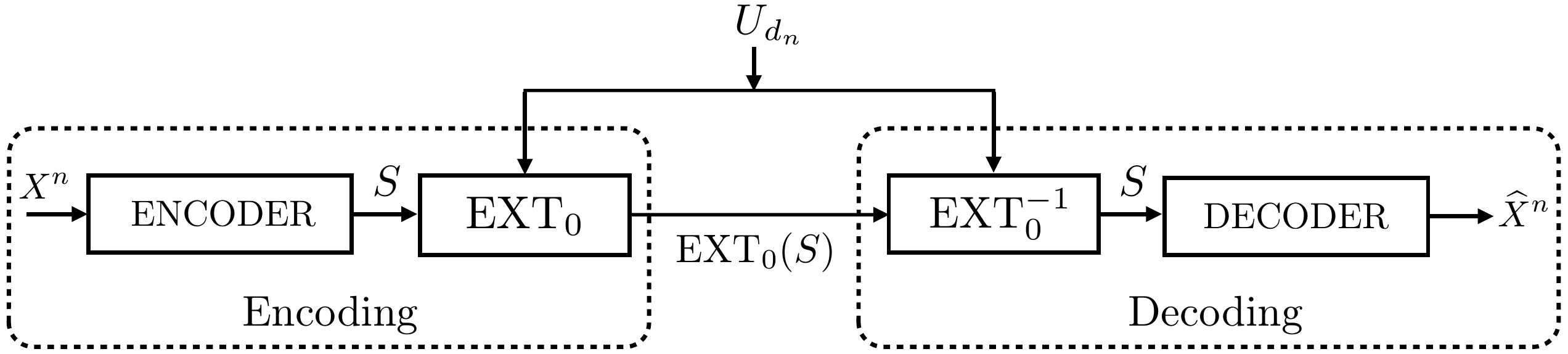}
  \caption{Encoding/decoding scheme for Proposition \ref{Extlossy}. The encoder/decoder is obtained from a typical-sequence based source code, and an invertible extractor  EXT$_0$.}
  \label{fig:inv_ext_scheme}
\end{figure}
\end{proposition}
\begin{proof}
See Appendix \ref{sec:extractors_uniform_compression}.
\end{proof}
Unfortunately, this scheme is not fully practical because it relies on a typical-sequence based compression. To provide at least one explicit and low-complexity example, we finally develop a uniform compression code based on polar codes for a binary memoryless source $(\mathcal{X},p_{X})$, $\mathcal{X} \triangleq \{ 0,1\}$. Let $\beta \in ]0,1/2[$, $n \in \mathbb{N}$, $N \triangleq 2^n$, and $\delta_N \triangleq 2^{-N^{\beta}}$. Let $G_N \triangleq  \left[ \begin{smallmatrix}
       1 & 0            \\[0.3em]
       1 & 1 
     \end{smallmatrix} \right]^{\otimes n}$ be the source polarization transform defined in \cite{Arikan10}, and set $A^N \triangleq X^N G_N$. For any set  $\mathcal{A} \triangleq \{ i_j \}_{j=1}^{|\mathcal{A}|}$ of indices in $\llbracket 1,N \rrbracket$, we define $A^N[\mathcal{A}] \triangleq \left[A_{i_1}, A_{i_2}, \ldots,A_{i_{|\mathcal{A}|}}\right]$. In the following, we denote the complement set operation by the superscript $c$. We also define the sets
\begin{align*}
\mathcal{V}_X  \triangleq &  \left\{ i\in \llbracket 1,N\rrbracket : {H} \left(A_i|A^{i-1} \right) > 1- \delta_N \right\} , \\
{\mathcal{H}_X} \triangleq  &\left\{ i\in \llbracket 1,N\rrbracket : {H}\left(A_i|A^{i-1}\right) >  \delta_N \right\}.
\end{align*}
A polar-based uniform compression code is obtained by defining the encoding function $\phi_N$ as follows. Let $U_{\card{\mathcal{H}_X \backslash \mathcal{V}_X}}$ denote a sequence of uniformly distributed random bits with length $|{{\mathcal{H}_X \backslash \mathcal{V}_X}}|$. Then,
\begin{align*}
  \phi_N (X^N\!,U_{\card{\mathcal{H}_X \backslash \mathcal{V}_X}}) \triangleq \! \left[A^N[\mathcal{V}_X],A^N[\mathcal{H}_X \! \backslash \mathcal{V}_X]\oplus U_{\card{\mathcal{H}_X \backslash \mathcal{V}_X}}\right]\!.
\end{align*}

\begin{proposition} \label{proppolar}
Let $(\mathcal{X},p_{X})$ be a binary memoryless source. Any rate $R>H(X)$ is achievable with a sequence of polar-based uniform compression codes such that the seed length $\card{\mathcal{H}_X \backslash \mathcal{V}_X}$ vanishes as the code length grows unbounded. In addition, the complexity of the encoding and decoding is $O(N \log N)$, where $N$ denotes the code length.
\end{proposition} 

\begin{proof}
See Appendix \ref{Apppol}.
\end{proof}

\subsection{Achievability of Theorem \ref{th:joint-source-channel} based on uniform compression codes} \label{secmodelu}


The uniform compression codes of Section \ref{Sec_uc} may now be combined with known wiretap codes (as depicted in Figure~\ref{fig:arch2}), whose properties we recall in the following lemma.
\begin{lemma}[Adapted from~{\cite[Proposition 1]{Bloch12}}]
  \label{lm:wiretap-codes}
  Consider a \ac{DMC} $(\calX,p_{YZ|X}$, $\calY\times\calZ$), in which a message $M_c\in\intseq{1}{2^{nR_c}}$ is encoded by means of a uniform auxiliary message $M_p\in\intseq{1}{2^{nR_p}}$. If there exists a joint distribution $p_{{Q}XYZ}$ that factorizes as $p_{Q}p_{X|{Q}}p_{YZ|X}$ such that
  \begin{align}
    R_c+R_p&<I(X;Y|{Q}) \label{eq:cond1}\\
    R_c&<I(X;Y|{Q})-I(X;Z|{Q}) \\
    R_p&>I(X;Z|{Q})\label{eq:cond3},
  \end{align}
 then there exists a sequence of wiretap codes $\{\calC_n\}_{n\geq 1}$ such that
 \begin{align*}
         \lim_{n\rightarrow\infty}&\max_{m}\mathbb{P}\left[\hat{M}_c\neq M_c|M_c=m\right]=0,\\
    \lim_{n\rightarrow\infty}&\max_m\mathbb{P}\left[\hat{M}_p\neq M_p|M_c=m\right]=0,\\
    \lim_{n\rightarrow\infty}&\max_{m}\V{p_{Z^n|M_c=m},p_{Z^n}}=0.
 \end{align*}
\end{lemma}
Let $\epsilon>0$. Going back again to the setting of Section~\ref{sec:channel-model}, we encode the confidential \ac{DMS} using a traditional source code as in Lemma~\ref{lm:compression}, and the public \ac{DMS} using a uniform compression code as in Proposition~\ref{Thlossya}. The corresponding source encoder-decoder pairs are denoted $(f_n^c,g_n^c)$ and $(f_n^p,g_n^p)$, respectively, and we set $M_c\triangleq f_n^c(V_c^n)\in\intseq{1}{2^{nR_c}}$ and $M_p\triangleq f_n^p(V_p^n,U_{d_n})\in\intseq{1}{2^{nR_p}}$.  We assume $n$ large enough so that 
\begin{align}
\P{(V_c^n,V_p^n)\neq (g_n^c(M_c),g_n^p(M_p,U_{d_n}))}\leq \epsilon,  \\
\V{p_{{M}_p},p_{U_{nR_p}}}<\epsilon. \label{eq:uMp}
\end{align}
 Under the conditions~\eqref{eq:cond1}-\eqref{eq:cond3} of Lemma~\ref{lm:wiretap-codes}, which are met whenever
\begin{align*}
  H(V_c)+H(V_p)&<I(X;Y|{Q}), \\ H(V_c) &< I(X;Y|{Q})-I(X;Z|{Q}),\\  H(V_p) &> I(X;Z|{Q}),
\end{align*}
 for $n$ sufficiently large there exists a wiretap code $\calC_n$ with encoder/decoder pair $(f_n,g_n)$ so that for any $m_c$, and for $\tilde{M}_p$ distributed according  to $p_{U_{nR_p}}$, the uniform distribution over $\llbracket 1, 2^{nR_p}\rrbracket$,  
 \begin{align} 
 	\mathbb{P}\left[\hat{\tilde{M}}_p\neq \tilde{M}_p|M_c=m_c\right] < \epsilon, \label{eq:wu}\\
 	\mathbb{P}\left[\hat{\tilde{M}}_c\neq {M}_c|M_c=m_c\right] < \epsilon,\label{eq:wu2}\\
 	\V{\tilde{p}_{Z^n|M_c=m_c},\tilde{p}_{Z^n}}\leq\epsilon,\label{eq:wu3}
 \end{align}
where $(\hat{\tilde{M}}_p,\hat{\tilde{M}}_c)$ is the estimate of $({\tilde{M}}_p,{{M}}_c)$ by the decoder of $\calC_n$, and for any $z^n$, $m_c$, $m_p$,
\begin{multline*}
  \tilde{p}_{Z^nM_cM_P}(z^n,m_c,m_p) \\ \eqdef   p_{Z^n|M_c=m_c,M_p=m_p}(z^n)p_{M_c}(m_c)p_{U_{nR_p}}(m_p).
\end{multline*}
Note that~\eqref{eq:wu}-\eqref{eq:wu3} holds by Lemma \ref{lm:wiretap-codes} because we have assumed $\tilde{M}_p$ uniformly distributed. We now study the consequences of using the wiretap code $\calC_n$ with $M_p$ (not exactly uniformly distributed) instead of $\tilde{M}_p$. Specifically, we note  $(\hat{{M}}_p,\hat{{M}}_c)$ the resulting estimate of $({{M}}_p,{{M}}_c)$ by the decoder of $\calC_n$, and define for any $z^n$, $m_c$, $m_p$,
\begin{multline*}
  p_{Z^nM_cM_P}(z^n,m_c,m_p) \\ \eqdef  p_{Z^n|M_c=m_c,M_p=m_p}(z^n)p_{M_c}(m_c)p_{M_p}(m_p).
\end{multline*}
We then have for any $m_c$,
\begin{align}
&\V{p_{Z^n|M_c=m_c},{p}_{Z^n}} \nonumber \\ \nonumber
& \stackrel{(a)}{\leq} \V{p_{Z^n|M_c=m_c},\tilde{p}_{Z^n|M_c=m_c}}+\V{\tilde{p}_{Z^n|M_c=m_c},\tilde{p}_{Z^n}} \\ \nonumber
& \phantom{---}+\V{\tilde{p}_{Z^n},p_{Z^n}} \\ \nonumber
& \stackrel{(b)}{\leq} \epsilon + \V{p_{Z^n|M_c=m_c},\tilde{p}_{Z^n|M_c=m_c}}+\V{\tilde{p}_{Z^n},p_{Z^n}} \\ \nonumber
& \stackrel{(c)}{\leq}  \epsilon + \sum_{z^n} \sum_{m_p} \left( p_{Z^n|M_c=m_c,M_p=m_p} (z^n)  \right. \\ \nonumber
& \phantom{---} \left. \times | p_{M_p}(m_p) -p_{U_{nR_p}}(m_p) | \right)  +\V{\tilde{p}_{Z^n},p_{Z^n}} \\ \nonumber
& = \epsilon + \V{p_{M_p},p_{U_{nR_p}}}+\V{\tilde{p}_{Z^n},p_{Z^n}}\\ \nonumber
& \stackrel{(d)}{\leq} 2 \epsilon  + \V{\tilde{p}_{Z^n},p_{Z^n}}\\ \nonumber
&\stackrel{(e)}{\leq} 2 \epsilon  +  \sum_{z^n} \sum_{m_c,m_p} \left( p_{M_c}(m_c)  p_{Z^n|M_c=m_c,M_p=m_p} (z^n) \right. \\ \nonumber
&\left. \phantom{---}\times | p_{M_p}(m_p) -p_{U_{nR_p}}(m_p) | \right)\\ \nonumber
& = 2 \epsilon + \V{p_{M_p},p_{U_{nR_p}}}\\
&\stackrel{(f)}{\leq} 3 \epsilon, \label{eq:disttildez}
\end{align}
where $(a)$, $(c)$, and $(e)$ follow by the triangle inequality, $(b)$ holds by \eqref{eq:wu3}, $(d)$ and $(f)$ hold by \eqref{eq:uMp}.

Consider then an optimal coupling~\cite{Aldous83} between $M_p$ and $\tilde{M}_p$ such that $\mathbb{P} [\mathcal{E}] = \mathbb{V}({p}_{M_p} ,p_{U_{nR_p}})$, where $\mathcal{E} \triangleq \{ M_p \neq \tilde{M}_p\}$. We have for any $m_c$,
\begin{align*}
&\mathbb{P}\left[\hat{M}_p\neq M_p|M_c=m_c\right] \\
&= \mathbb{P}\left[\hat{M}_p\neq M_p|M_c=m_c ,\mathcal{E}^c \right] \mathbb{P}\left[\mathcal{E}^c \right] \\
& \phantom{---}+ \mathbb{P}\left[\hat{M}_p\neq M_p|M_c=m_c ,\mathcal{E} \right] \mathbb{P}\left[\mathcal{E} \right] \\
& \leq  \mathbb{P}\left[\hat{M}_p\neq M_p|M_c=m_c ,\mathcal{E}^c \right]  + \mathbb{P}\left[\mathcal{E} \right] \\
& =   \mathbb{P}\left[\hat{M}_p\neq M_p|M_c=m_c ,\mathcal{E}^c \right] + \mathbb{V}({p}_{M_p} ,p_{U_{nR_p}}) \\
& =   \mathbb{P}\left[\hat{\tilde{M}}_p\neq \tilde{M}_p|M_c=m_c\right] + \mathbb{V}({p}_{M_p} ,p_{U_{nR_p}}) \\
&\leq 2 \epsilon,
\end{align*}
where the last inequality follows from \eqref{eq:uMp} and \eqref{eq:wu}. Similarly, using \eqref{eq:uMp} and \eqref{eq:wu2}, we have for any $m_c$,
\begin{align*}
\mathbb{P}\left[\hat{M}_c\neq M_c|M_c=m_c\right] \leq 2 \epsilon.
\end{align*}

Encoding the sources into codewords with $\mathcal{C}_n$ as $f_n(f_n^c(V_c^n),f_n^p(V_p^n,U_{d_n}))$, and forming estimates from the channel output $Y^n$ as $\hat{V}_c^n \triangleq g_n^c(g_n(Y^n))$, and $\hat{V}_p^n \triangleq g_n^p(g_n(Y^n),U_{d_n})$, we obtain again 
\begin{align*}
&\P{(V_c^n,V_p^n)\neq (\hat{V}_c^n,\hat{V}_p^n)}\\
&\leq \P{(V_c^n,V_p^n)\neq (\hat{V}_c^n,\hat{V}_p^n) | (\hat{{M}}_p,\hat{{M}}_c)= ({{M}}_p,{{M}}_c) } \\
& \phantom{---}+ \P{(\hat{{M}}_p,\hat{{M}}_c) \neq({{M}}_p,{{M}}_c) } \\
&\leq  5\epsilon.
\end{align*}
For any $v_c^n\in\mathcal{V}_c^n$, we also have
\begin{align*}
  &\V{p_{Z^n|V_c^n=v_c^n},p_{Z^n}}  \\
  &\stackrel{(a)}{\leq}\sum_{m_c}p_{M_c|V_c^n=v_c^n}(m_c)\V{p_{Z^n|M_c=m_c,V_c^n=v_c^n},p_{Z^n}}\\
  &\stackrel{(b)}{=}\sum_{m}p_{M_c|V_c^n=v_c^n}(m_c)-\V{p_{Z^n|M_c=m_c},p_{Z^n}}\\
  &\leq 3\epsilon,
\end{align*}
where $(a)$ follows by the triangle inequality, $(b)$ holds because $Z^n \to M_c \to V_c^n$. Since $\epsilon>0$ can be chosen arbitrarily small, we obtain again the achievability part of Theorem~\ref{th:joint-source-channel}.

\section{Conclusion}
\label{sec:conclusion}
We have proposed and analyzed two coding architectures for multiplexing confidential and public messages and achieve information-theoretic secrecy over the wiretap channel. Our first architecture relies on wiretap codes that do not require uniform randomization, while the second architecture exploits compression codes that output nearly uniform messages. By showing that secrecy can be achieved with only vanishing-rate randomness resources, and without reducing the overall rate of reliable communication, the proposed architectures establish that secrecy can be achieved at negligible cost. 
 
An important issue that we have not addressed is the design of \emph{universal} wiretap codes that merely require that the public message carries enough randomness, and do not require the knowledge of the statistics. Some results in this direction are already available in~\cite{Hayashi12}. Finally, the design of actual codes for the proposed architecture remains an important avenue for future research.

\appendices

\section{Converse of Theorem \ref{th:joint-source-channel}} \label{App_converse}
We consider the problem described in Section \ref{sec:channel-model} when the uniformly distributed seed is shared between the encoder and the decoder, as it is the case in Section \ref{Sec_Sol2}. Obviously, the converse will also hold when the seed is not available at the decoder, as it is the case in Section \ref{Sec_Sol1}. We develop our converse following techniques similar to Csisz\'ar and K\"orner~\cite{Csiszar78} and Oohama and Watanabe \cite{Watanabe12}. Although the ideas are similar, the converse does not follow directly from these known results because of the presence of a seed with length $d_n$. Formally, consider two sources $(\mathcal{V}_c,p_{V_c})$ and $(\mathcal{V}_p,p_{V_p})$ that can be transmitted reliably and secretly. Then, there exists a code with block length $n$ such that
\begin{align}
P((\hat{V}_c^n,\hat{V}_p^n)\neq ({V}_c^n,{V}_p^n))& \leq \epsilon'_n  \text{  (reliability),}  \label{eq:reliability} \\
I(V_c^n;Z^n) &\leq \delta_n   \text{  (secrecy),} \label{eq:secrecy}\\
d_n/n &\leq \mu_n  \text{  (sub-linear seed rate),}  \label{eq:key}
\end{align}
where $\lim_{n\to \infty} \epsilon'_n = \lim_{n\to \infty} \delta_n =\lim_{n\to \infty} \mu_n =0$. We also define $\epsilon_n = \epsilon'_n + 1/n$.
Consequently,
\begin{align}
  H(V_c^n)& \stackrel{(a)}{=}  H(V_c^nV_p^n)-  H(V_p^n) \nonumber \\
  &=I(V_c^nV_p^n;Y^nU_{d_n})+H(V_p^nV_c^n|Y^nU_{d_n})-H(V_p^n)\nonumber \\
  &\stackrel{(b)}{\leq} I(V_c^nV_p^n;Y^nU_{d_n}) +n\epsilon_n-H(V_p^n ) \nonumber \\
    &\leq I(V_c^nV_p^n;Y^nU_{d_n}) +n\epsilon_n- I(V_p^n;Z^n|V_c^n ) \nonumber \\
  &\stackrel{(c)}{\leq}I(V_c^nV_p^n;Y^nU_{d_n}) +n\epsilon_n -I(V_p^n V_c^n;Z^n) + n\delta_n  \nonumber \\
  &\leq I(V_c^nV_p^n;Y^n)\! -I(V_p^n V_c^n;Z^n) +n\epsilon_n + n\delta_n +d_n, \label{eq:c1}
\end{align}
where $(a)$ holds by the {independence of the sources}, $(b)$ holds by \eqref{eq:reliability} and Fano's inequality, $(c)$ holds by \eqref{eq:secrecy}. Next,
\begin{align}
  H(V_p^n)+n\mu_n &\geq H(V_p^n)+ d_n \nonumber\\ \nonumber
                &\stackrel{(a)}{=} H(V_p^nU_{d_n})\\ \nonumber
  &\stackrel{(b)}{=} H(V_p^nU_{d_n}|V_c^n)\\ \nonumber
  &\stackrel{(c)}{\geq} H(X^n|V_c^n)\\ \nonumber
  &\geq I(X^n;Z^n|V_c^n)&\\ \nonumber
  &\stackrel{(d)}{\geq} I(X^nV_c^n;Z^n)-n\delta_n \\
  &\stackrel{(e)}{\geq} I(X^n;Z^n)-n\delta_n,
\end{align}
where $(a)$ and $(b)$ hold by independence of the sources and the seed, $(c)$ holds because $X^n$ is a function of $V_p^n,U_{d_n},V_c^n$, $(d)$ holds by \eqref{eq:secrecy}, $(e)$ holds because $V_c^n-X^n-Z^n$ forms a Markov chain. Similarly,
\begin{align}
  H(V_c^n)+H(V_p^n)& \stackrel{(a)}{=} H(V_p^nV_c^n|U_{d_n})\nonumber\\ \nonumber
  &=I(V_p^n V_c^n;Y^n|U_{d_n})+H(V_p^nV_c^n|Y^nU_{d_n})&\\ \nonumber
  &\stackrel{(b)}{\leq} I(V_p^n V_c^n;Y^n|U_{d_n}) +n\epsilon_n\\
  &\leq I(V_p^n V_c^nU_{d_n};Y^n) +n\epsilon_n,
\end{align}
where $(a)$ holds by independence of the sources and the seed, $(b)$ holds by \eqref{eq:reliability} and Fano's inequality.
Finally,
\begin{align}
  n\mu_n&\stackrel{(a)}{\geq} {d_n}  \nonumber \\ \nonumber
  & \stackrel{(b)}{=} H(U_{d_n})\\ \nonumber
  &\geq H(U_{d_n}|V_c^nV_p^n)\\ \nonumber
  &\stackrel{(c)}{\geq} H(X^n|V_c^nV_p^n)\\
  &\geq I(X^n;Z^n|V_c^nV_p^n),
\end{align}
where $(a)$ holds by \eqref{eq:key}, $(b)$ holds by uniformity of the seed, $(c)$ holds because $X^n$ is a function of $V_p^n,U_{d_n},V_c^n$. The single letterization is obtained by introducing a random variable $I$ uniformly distributed over $\llbracket 1,n \rrbracket$ and defining
\begin{align*}
 {Q}_i&=(Y_1^{i-1},Z_{i+1}^n),\quad  V_i=({Q}_i,V_c^n,V_p^n),\\ {Q}&=({Q}_I,I),\quad V=(V_I,I),\\
  X &= X_I,\quad Y=Y_I,\quad Z=Z_I.
\end{align*}
Note that the joint distribution of ${Q},V,X,Y,Z$ factorizes as
$\label{eq:mc}  p_{Q}p_{V|{Q}}p_{X|V}W_{YZ|X}.$
Then, using Csisz\'ar's sum-equality
\begin{align}
  &I(V_c^nV_p^n;Y^n) -I(V_p^n V_c^n;Z^n) \nonumber \\
  &\leq \sum_{i=1}^n\left[I(V_c^nV_p^n;Y_i|Y_{1}^{i-1}) -I(V_p^n V_c^n;Z_i|Z_{i+1}^n)\right] \nonumber\\ \nonumber
  &= \sum_{i=1}^n\left[I(V_c^nV_p^n;Y_i|Y_{1}^{i-1}Z_{i+1}^n) -I(V_p^n V_c^n;Z_i|Y_{1}^{i-1}Z_{i+1}^n)\right]\\
  &=n[I(V;Y|{Q})-I(V;Z|{Q})].
\end{align}
In addition, 
\begin{align}
  & I(X^n;Z^n|V_c^nV_p^n) \nonumber \\
  &= \sum_{i=1}^n\left[H(Z_i|Z_{1}^{i+1}V_c^nV_p^n)-H(Z_i|Z_{1}^{i+1}X^nV_c^nV_p^n)\right] \nonumber \\ \nonumber
  &\geq \!\!\sum_{i=1}^n\!\left[H(Z_i|Y_1^{i-1}Z_{1}^{i+1}V_c^nV_p^n)\!-\!H(Z_i|Z_{1}^{i+1}Y_1^{i-1}X_iV_c^nV_p^n\right]\\ \nonumber
  &=\sum_{i=1}^n I(X_i;Z_i|Z_{1}^{i+1}Y_1^{i-1}V_c^nV_p^n)\\
  &=nI(X;Z|V),
\end{align}
where the inequality holds because  $Z_i-X_i-Z_{1}^{i+1}Y_1^{i-1}X^nV_c^nV_p^n$ forms a Markov chain. Similarly,
\begin{align}
  I(X^n;Z^n) &=\sum_{i=1}^n\left(H(Z_i|Z_{i+1}^n)-H(Z_i|Z_{i+1}^nX^n)\right) \nonumber \\ \nonumber
             &=\sum_{i=1}^n\left(H(Z_i|Z_{i+1}^nY_1^{i-1})-H(Z_i|X_iZ_{i+1}^nY_1^{i-1})\right)\\ \nonumber
             &=\sum_{i=1}^nI(X_i;Z_i|{Q}_i)\\
             &=nI(X;Z|{Q}).
\end{align}
Finally,
\begin{align}
  &I(V_p^n V_c^nU_{d_n};Y^n)\nonumber \\
  &= \sum_{i=1}^nI(V_p^n V_c^nU_{d_n};Y_i|Y_{1}^{i-1}) \nonumber\\ \nonumber
  &\leq \sum_{i=1}^nI(V_p^n V_c^nU_{d_n} Y_{1}^{i-1}Z_{i+1}^n;Y_i)\\ \nonumber
  &= \sum_{i=1}^nI(X_i Y_{1}^{i-1}Z_{i+1}^n;Y_i)\\ \nonumber
  &=\sum_{i=1}^nI(X_i{Q}_i;Y_i)\\   \nonumber
  &=nI(X{Q};Y)\\
  &=nI(X;Y|{Q}), \label{eq:c2}
\end{align}
where the second equality holds because $X_i$ is a function of $V_p^n,U_{d_n},V_c^n$ and  $Y_i- Z_{i+1}^{n}Y_1^{i-1}X_i - V_c^nV_p^n U_{d_n}$ forms a Markov chain. Combining~\eqref{eq:c1} -- \eqref{eq:c2} we obtain
\begin{align*}
  H(V_c)&\leq I(V;Y|{Q})-I(V;Z|{Q})+\epsilon_n+\delta_n+\mu_n\\
  H(V_p)+\mu_n &\geq I(X;Z|{Q})\\
  H(V_c)+H(V_p)&\leq I(X;Y|{Q})+\epsilon_n\\
  \mu_n &\geq I(X;Z|V).
\end{align*}
Note that using $p_{QVXYZ} = p_{Q}p_{V|{Q}}p_{X|V}W
p_{YZ|X}$ we have
\begin{align*}
  I(V;Z|{Q})&=I(VX;Z|{Q})-I(X;Z|{Q}V)\\
  &=I(X;Z|{Q})+I(V;Z|{Q}X)-I(X;Z|V)\\
  &\geq I(X;Z|{Q})-\mu_n,
\end{align*}
and
\begin{align*}
  I(V;Y|{Q})
  &\leq I(VX;Y|{Q})\\
  &=I(X;Y|{Q})+I(V;Y|{Q}X)\\
  &=I(X;Y|{Q}).
\end{align*}
Hence, we must have
\begin{align*}
    H(V_c)&\leq I(X;Y|{Q})-I(X;Z|{Q}) \!+\!\epsilon_n\!+\!\delta_n\!+\!2\mu_n,\\
  H(V_p)&\geq I(X;Z|{Q})-\mu_n,\\
  H(V_c)+H(V_p)&\leq I(X;Y|{Q})+\epsilon_n.  
\end{align*}

\section{Proof of Proposition \ref{prop_non_uniform2}} 
\label{App_renyi_constraint}
We fix a joint distribution $p_{{Q}X}$ on $\mathcal{{Q}}\times\mathcal{X}$ such that\footnote{If such a probability distribution does not exist the result of Lemma~\ref{prop_non_uniform2} is trivial and there is nothing to prove.} $I(X;Z|{Q}) \leq \lim_{n\rightarrow\infty}\frac{1}{n}H_2(M_p)$ and $I(X;Y|{Q})-I(X;Z|{Q})>0$. Let $\epsilon>0$, $R_0>0$, and $n\in\mathbb{N}$. We randomly construct a sequence of codes $\{C_n\}_{n\in \mathbb{N}}$ as follows. We generate $2^{nR_0}$ sequences independently at random according to $p_{Q}$, which we label $q^n(i)$ for $i\in \llbracket 1, 2^{nR_0} \rrbracket$. For each sequence $q^n(i)$, we generate $2^{n(R_c+R_p)}$ sequences independently a random according to $p_{X|{Q}}$, which we label $x^{n}(i,j,s)$ with $j\in \llbracket 1, 2^{nR_c} \rrbracket $ and $s\in \llbracket 1,2^{nR_p}\rrbracket$. To transmit a message $i\in \llbracket 1,2^{nR_0}\rrbracket$ and $j\in \llbracket 1,2^{nR_c} \rrbracket$, the transmitter obtains a realization $s$ of the public message $M_p\in \llbracket 1,2^{nR_p} \rrbracket$, and transmits $x^n(i,j,s)$ over the channel. Upon receiving $y^n$, Bob decodes $i$ as the received index if it is the unique one such that $(q^n(i),y^n)\in\mathcal{T}_{\epsilon}^n({Q}Y)$; otherwise he declares an error. Bob then decodes $(j,s)$ as the other pair of indices if it is the unique one such that $(q^n(i),x^n(i,j,s),y^n)\in\mathcal{T}_{\epsilon}^n({Q}XY)$. Similarly, upon receiving $z^n$, Eve decodes $i$ as the received index if it is the unique one such that $(q^n(i),z^n)\in\mathcal{T}_{\epsilon}^n({Q}Z)$; otherwise she declares an error. For a particular code $C_n$, we note $\mathbf{P}_e(C_n)$ the probability that Bob does not recover correctly $(i,j,s)$ and that Eve does not recover correctly $i$.

\begin{lemma}
  \label{lm:reliability}
  If $R_0<I({Q};Y)$ and $R_c+R_p<I(X;Y|{Q})$, then $\mathbb{E} [ \mathbf{P}_e(C_n) ]\leq 2^{-\alpha n}$ for some $\alpha>0$.
\end{lemma}
\begin{proof}
  The proof follows from a standard random coding argument and is omitted.
\end{proof}
\begin{lemma}
  \label{lm:secrecy}
   If $\displaystyle\lim_{n\rightarrow\infty}\frac{1}{n}H_2(M_p)>I(X;Z|{Q})$, then we have $\mathbb{E}_{C_n} \left[\mathbb{V}(p_{M_cZ^n},p_{M_c}p_{Z^n}) \right]\leq 2^{-\beta n}$ for some $\beta >0$ and all $n\in\mathbb{N}$ sufficiently large.
\end{lemma}
\begin{proof}
The proof relies on a careful analysis and modification of the ``cloud-mixing'' lemma~\cite{Cuff09}. Let $\epsilon>0$. For clarity, we denote here $\hat{p}_{{Q}^nX^nZ^n}$ the joint distribution of $({Q}^n,X^n,Z^n)$  induced by the code, as opposed to $p_{{Q}^nX^nZ^n}$ defined as
  \begin{align*}
    p_{{Q}^nX^nZ^n}(q^n,x^n,z^n) = {p}_{Z^n|X^n}(z^n|x^n)p_{X^n{Q}^n}(x^n,q^n).
  \end{align*}
First note that the variational distance $\mathbb{V}(\hat{p}_{M_cZ^n},p_{M_c}\hat{p}_{Z^n})$ can be bounded as follows.
  \begin{align*}
    &\mathbb{V}(\hat{p}_{M_cZ^n},p_{M_c}\hat{p}_{Z^n})\nonumber \\
    &\quad\leq \mathbb{V}(\hat{p}_{M_c{Q}^nZ^n},p_{M_c}\hat{p}_{{Q}^nZ^n})\\
    &\quad =\mathbb{E}_{{Q}^nM_c} \left[\mathbb{V}(\hat{p}_{Z^n|M_c{Q}^n},\hat{p}_{Z^n|{Q}^n})\right]\\
    &\quad \leq \mathbb{E}_{{Q}^nM_c} \left[\mathbb{V}(\hat{p}_{Z^n|M_c{Q}^n},p_{Z^n|{Q}^n})+\mathbb{V}(p_{Z^n|{Q}^n},\hat{p}_{Z^n|{Q}^n})\right]\\
    &\quad \leq 2 \mathbb{E}_{{Q}^nM_c} \left[\mathbb{V}(\hat{p}_{Z^n|M_c{Q}^n},p_{Z^n|{Q}^n}) \right]
  \end{align*}
Then, let ${Q}_1^n$ be the sequence in $\mathcal{{Q}}^n$ corresponding to $M_0=1$. By symmetry of the random code construction, the average of the variational distance $\mathbb{V}(\hat{p}_{M_cZ^n},p_{M_c}\hat{p}_{Z^n})$ over randomly generated codes $C_n$ satisfies
  \begin{multline*}
    \mathbb{E}_{C_n} \left[\mathbb{V}(\hat{p}_{M_cZ^n},p_{M_c}\hat{p}_{Z^n}) \right] \\ \leq 2\mathbb{E}_{C_n} \left[ \mathbb{V}(\hat{p}_{Z^n|{Q}^n={Q}_1^nM_c=1},p_{Z^n|{Q}^n={Q}_1^n}) \right],
  \end{multline*}
  where
  \begin{align*}
    \hat{p}_{Z^n|{Q}^n={Q}_1^nM_c=1}(z^n)=\sum_{k=1}^{2^{nR_p}}{p}_{Z^n|X^n}(z^n|x^n(1,1,k))p_{M_p}(k).
  \end{align*}
The average over the random codes can be split between the average of ${Q}_1^n$ and the random code $C_n(q_1^n)$ for a fixed value of $q_1^n$, so that
  \begin{multline*}
    \mathbb{E}_{C_n} \left[\mathbb{V}(\hat{p}_{Z^n|{Q}^n={Q}_1^nM_c=1},p_{Z^n|{Q}^n={Q}_1^n}) \right] \\
    \begin{split}
      &= \sum_{q_1^n\in\mathcal{{Q}}^n}\! p_{{Q}^n}(q_1^n)\mathbb{E}_{C_n(q_1^n)} \left[\mathbb{V}(\hat{p}_{Z^n|{Q}^n=q_1^nM_c=1},p_{Z^n|{Q}^n=q_1^n}) \right]\\
      &\leq \!\!\! \sum_{q_1^n\in\mathcal{T}_{\epsilon}^n{(U)}}\!\!\!\!\!p_{U^n}(q_1^n)\mathbb{E}_{C_n(q_1^n)}\left[\mathbb{V}(\hat{p}_{Z^n|{Q}^n=q_1^nM_c=1},p_{Z^n|{Q}^n=q_1^n})      \right]\\
      & \phantom{---} +2\mathbb{P} \left[{Q}^n\notin \mathcal{T}_{\epsilon}^n({Q})   \right],      
    \end{split}
  \end{multline*}
  where the last inequality follows from the fact that the variational distance is always less than 2. The first term on the right-hand side vanishes exponentially with $n$, and we now proceed to bound the expectation in the second term following~\cite{Cuff09}. First note that, for any $z^n\in\mathcal{Z}^n$,
  \begin{align*}
&\mathbb{E}_{C_n(q_1^n)} \left[\hat{p}_{Z^n|{Q}^n=q_1^nM_c=1}(z^n) \right]\\
&=\mathbb{E}_{C_n(q_1^n)} \left[ \sum_{k=1}^{2^{nR_p}}p_{Z^n|X^n}(z^n|x^n(1,1,k))p_{M_p}(k) \right]\\
&=\sum_{k=1}^{2^{nR_p}} \mathbb{E}_{C_n(q_1^n)} \left[ p_{Z^n|X^n}(z^n|x^n(1,1,k)) \right] p_{M_p}(k)\\
&=p_{Z^n|{Q}^n=q_1^n}(z^n).
  \end{align*}
  We now let $\mathds{1}$ denote the indicator function and we define
  \begin{align*}
    &p^{(1)}(z^n) \triangleq\sum_{k=1}^{2^{nR_p}}{p}_{Z^n|X^n}(z^n|x^n(1,1,k))p_{M_p}(k)   \\
    & \phantom{--------} \times\mathds{1}\{(x^n(1,1,k),z^n)\in\mathcal{T}_{2\epsilon}^n(XZ|q_1^n)\},\\
    &p^{(2)}(z^n) \triangleq\sum_{k=1}^{2^{nR_p}}{p}_{Z^n|X^n}(z^n|x^n(1,1,k))p_{M_p}(k)     \\
    & \phantom{--------} \times\mathds{1}\{(x^n(1,1,k),z^n)\notin\mathcal{T}_{2\epsilon}^n(XZ|q_1^n)\},
  \end{align*}
 so that we can upper bound $\mathbb{V}(\hat{p}_{Z^n|{Q}^n=q_1^nM_c=1},p_{Z^n|{Q}^n=q_1^n})$ as
 \begin{align}
   &\mathbb{V}(\hat{p}_{Z^n|{Q}^n=q_1^nM_c=1},p_{Z^n|{Q}^n=q_1^n})  \nonumber \\
   &\leq \sum_{z^n\notin\mathcal{T}_{2\epsilon}^n{(Z|q_1^n)}} \left| \hat{p}_{Z^n|{Q}^n=q_1^nM_c=1}(z^n)-p_{Z^n|{Q}^n=q_1^n}(z^n)\right| \label{eq:1}\\
   &\qquad +\sum_{z^n\in\mathcal{T}_{2\epsilon}^n{(Z|q_1^n)}}\left|p^{(1)}(z^n)-\mathbb{E} \left[ p^{(1)}(z^n) \right] \right|\label{eq:2}\\
   &\qquad +\sum_{z^n\in\mathcal{T}_{2\epsilon}^n{(Z|q_1^n)}}\left|p^{(2)}(z^n)-\mathbb{E} \left[ p^{(2)}(z^n) \right]\right|\label{eq:3}.
 \end{align}
 Taking the expectation of the term in~\eqref{eq:1} over $C_n(q_1^n)$, we obtain
 \begin{multline*}
   \mathbb{E} \left[ \sum_{z^n\notin\mathcal{T}_{2\epsilon}^n(Z|q_1^n)} \left| \hat{p}_{Z^n|{Q}^n=q_1^nM_c=1}(z^n)-p_{Z^n|{Q}^n=q_1^n}(z^n) \right| \right]\\
   \begin{split}
     &\leq \sum_{z^n\notin\mathcal{T}_{2\epsilon}^n(Z|q_1^n)} \mathbb{E} \left[\hat{p}_{Z^n|{Q}^n=q_1^nM_c=1}(z^n)+ p_{Z^n|{Q}^n=q_1^n}(z^n)\right]\\
     & =     2\sum_{z^n\notin\mathcal{T}_{2\epsilon}^n(Z|q_1^n)}p_{Z^n|{Q}^n=q_1^n}(z^n),
   \end{split}
 \end{multline*}
 which vanishes exponentially fast as $n$ goes to infinity for $q_1^n\in\mathcal{T}_\epsilon^n({Q})$. Similarly, taking the expectation of the term in~\eqref{eq:3} over $C_n(q_1^n)$, we obtain
\begin{multline*}
  \mathbb{E} \left[ \sum_{z^n\in\mathcal{T}_{2\epsilon}^n(Z|q_1^n)} \left| p^{(2)}(z^n)-\mathbb{E}\left[ p^{(2)}(z^n) \right] \right| \right]\\
  \begin{split}
    &\leq\mathbb{E} \left[ \sum_{z^n\in\mathcal{Z}^n}  \left| p^{(2)}(z^n)-\mathbb{E} \left[p^{(2)}(z^n) \right]\right| \right]\\
    &\leq 2 \sum_{z^n\in\mathcal{Z}^n}\mathbb{E} \left[ p^{(2)}(z^n) \right]\\
    &=\sum_{z^n\in\mathcal{Z}^n}\mathbb{E}\left[{p}_{Z^n|X^n}(z^n|X^n(1,1,1) ) \right.\\
    & \left.\phantom{----} \times \mathds{1}\{(X^n(1,1,1),z^n)\notin\mathcal{T}_{2\epsilon}^n(XZ|q_1^n)\} \right]\\
    &=\sum_{(x^n,z^n)\notin\mathcal{T}_{2\epsilon}^n(XZ|q_1^n)}p_{Z^nX^n|{Q}^n=q_1^n}(z^n,x^n),
  \end{split}
\end{multline*}
which vanishes exponentially fast with $n$. Finally, we focus on the expectation of the term in~\eqref{eq:2} over $C_n(q_1^n)$. For $z^n\in \mathcal{T}_{2\epsilon}^n(Z|q_1^n)$, Jensen's inequality and the concavity of $x\mapsto \sqrt{x}$ guarantee that
 \begin{align*}
   \mathbb{E} \left[\left|p^{(1)}(z^n)-\mathbb{E} \left[p^{(1)}(z^n) \right] \right| \right] \leq \sqrt{\textup{Var}\left(p^{(1)}(z^n)\right)}.
 \end{align*}
 In addition,
 \begin{align*}
   &\textup{Var}\left(p^{(1)}(z^n)\right)=
   \sum_{k=1}^{2^{nR_p}}p_{M_p}(k)^2\text{Var}\left( {p}_{Z^n|X^n}(z^n|X^n(1,1,k)) \right. \\
    & \left. \phantom{-----l---} \times\mathds{1}\{(X^n(1,1,k),z^n)\in\mathcal{T}_{2\epsilon}^n(XZ|q_1^n)\}\right)
 \end{align*}
 Note that
 \begin{multline*}
 \begin{split}
&\text{Var}\left({{p}_{Z^n|X^n}(z^n|X^n(1,1,k))} \right.\\
    & \left. \phantom{-----} \times \mathds{1}\{(X^n(1,1,k),z^n)\in\mathcal{T}_{2\epsilon}^n{(XZ|q_1^n)}\}\right)\\
  & \leq \sum_{x^n\in\mathcal{X}^n}p_{X^n|{Q}^n=q_1^n}(x^n)
  \left({{p}_{Z^n|X^n}(z^n|x^n)}\right.\\
    & \left. \phantom{-----} \times\mathds{1}\{(x^n,z^n)\in\mathcal{T}_{2\epsilon}^n(XZ|q_1^n)\}\right)^2\\
  &=\sum_{x^n:(x^n,z^n)\in\mathcal{T}_{2\epsilon}^n(XZ|q_1^n)}p_{X^n|{Q}^n=q_1^n}(x^n) {p}_{Z^n|X^n}(z^n|x^n)^2\\
  &\stackrel{(a)}{\leq}2^{-n(H(Z|X)-\delta(\epsilon))} \\
    &  \phantom{---} \times \!\!\!\! \sum_{x^n:(x^n,z^n)\in\mathcal{T}_{2\epsilon}^n(XZ|q_1^n)} \!\!\!\!\!\!\!\!\!\!\!\!p_{X^n|{Q}^n=q_1^n}(x^n) {p}_{Z^n|X^n}(z^n|x^n)\\
  &\leq 2^{-n(H(Z|X)-\delta(\epsilon))} p_{Z^n|{Q}^n=q_1^n}(z^n)\\
  &\stackrel{(b)}{\leq}2^{-n(H(Z|X)+H(Z|{Q})-\delta(\epsilon))},
\end{split}
\end{multline*}
where $(a)$ and $(b)$ follow from the AEP; therefore, 
 \begin{align*}   
   &\textup{Var} \left(p^{(1)}(z^n) \right)\\
   &\leq 2^{-n(H(Z|X)+H(Z|{Q})-\delta(\epsilon))}\sum_{k=1}^{2^{nR_p}}p_{M_p}(k)^2\\
   &\leq 2^{-n(H(Z|X)+H(Z|{Q})-\delta(\epsilon))+\frac{H_2(M_p)}{n}}.
 \end{align*}
 and
 \begin{align*}
     &\sum_{z^n\in\mathcal{T}_{2\epsilon}^n{(Z|q_1^n)}}\mathbb{E} \left[\left| p^{(1)}(z^n)-\mathbb{E}\left[p^{(1)}(z^n) \right] \right| \right]\nonumber \\
     &\leq 2^{nH(Z|{Q})}2^{-\frac{n}{2}(H(Z|X)+H(Z|{Q})-\delta(\epsilon)+\frac{H_2(M_p)}{n})}\\
     &= 2^{-\frac{n}{2}(\frac{H_2(M_p)}{n}-I(X;Z|{Q})-\delta(\epsilon))}
 \end{align*}
 Hence, if $\lim_{n\rightarrow\infty}\frac{1}{n}H_2(M_p)>I(X;Z|{Q})+\delta(\epsilon)$, the sum vanishes as $n$ goes to infinity, which concludes the proof. 
 \end{proof}
We point out that a generalized version of Lemma~\ref{lm:secrecy} may now be found in~\cite{Hayashi12}; in fact,~\cite[Theorem 14]{Hayashi12} develops a general exponential bound on the secrecy metric, and a close inspection of their result shows a tighter exponent involves the R\'enyi entropy of order $1+\rho$ with $\rho\in[0,1]$ in place of the R\'enyi entropy of order 2. Actually, \cite{parizi2016exact} shows that this is the best exponent with random codes.

Using Markov's inequality, we conclude that there exists at least one code ${C}_n$ satisfying the rate inequalities in Lemma~\ref{lm:reliability} and Lemma~\ref{lm:secrecy}, such that $\mathbf{P}_e({C}_n)\leq 3\cdot 2^{-\alpha n}$ and $\mathbb{V}(p_{M_cM_0Z^n},p_{M_c}p_{M_0Z^n})\leq 3\cdot 2^{-\beta n}$. 
We now define
\begin{align*}
  &P_1(m) \eqdef \P{M_c\neq \hat{M}_c|M_c=m},\\
  &   P_2(m) \eqdef \P{M_p\neq \hat{M}_p|M_c=m},\\
  & S(m) = \mathbb{V}(p_{Z^n|M_c=m},p_{Z^n}).
\end{align*}
Since $\E{P_1(M_c)}\leq 2^{-\alpha n}$, $\E{P_2(M_c)}\leq 2^{-\alpha n}$, and $\E{S(M_c)}\leq 2^{-\beta n}$, we conclude with Markov's inequality that for $n$ large enough, we have
\begin{align*}
  P_1(m)<2^{-\alpha n+2}, \quad P_2(m) <2^{-\alpha n+2},\quad S(m)< 2^{-\beta n+2}
\end{align*}
for at least a quarter of the messages $m$. Expurgating  the code ${C}_n$ to retain only these messages concludes the proof.

\section{Proof of Proposition \ref{Thlossya}} \label{App_seed}
\subsection{Achievability} \label{secdef}
We show next that there exists a sequence of $(2^{nR},n,2^{d_n})$ uniform compression codes $\{ \mathcal{C}_n \}_{n \in \mathbb{N}^*}$ such that $H(X)$ is achievable with a seed length $d_n$ scaling as
\begin{align*}
d_n & = \Theta (\upsilon_n \sqrt{n}   ), \text{ for any } \{\upsilon_n\}_{n \in \mathbb{N}} \text{ with } \lim_{n \to \infty} \upsilon_n = +\infty.
\end{align*}

Let $\epsilon_1>0$, $\epsilon >0$, $n\in \mathbb{N}$, $d_n \in \mathbb{N}$, $R>0$. Define $M_n \triangleq 2^{nR}$ and $\mathcal{M}_n \triangleq \llbracket 1 , M_n \rrbracket$. Consider a random mapping $\Phi : \mathcal{X}^n \times \mathcal{U}_{d_n} \to\mathcal{M}_n$, and its associated decoder  $\Psi : \mathcal{M}_n \times \mathcal{U}_{d_n} \to \mathcal{X}^n$. Given $(m,u_{d_n}) \in \mathcal{M}_n \times \mathcal{U}_{d_n}$, the decoder outputs $\hat{{x}}^n$ if it is the unique sequence such that $\hat{{x}}^n \in \mathcal{T}_{\epsilon_1}^n({X})$ and $\Phi(\hat{{x}}^n,u_{d_n})=m$; otherwise it outputs an error.  We let $M \triangleq \Phi({X}^n,U_{d_n})$, and define $\mathbf{P}_e \triangleq \mathbb{P}[{X}^n \neq \Psi  (\Phi ({X}^n,U_{d_n}),U_{d_n})] $, $\mathbf{U}_e \triangleq \mathbb{V} \left(p_M, p_{\mathcal{U}_{ M_n}}\right)$.
\begin{itemize}
\item We first determine a condition over $R$ to ensure $\mathbb{E}_{\Phi} \left[\mathbf{U}_e\right] \leq \epsilon$. Note that $\forall m \in \mathcal{M}_n$,
$$
 p_{M}(m) = \sum_{{x}^n} \sum_u p(x^n,u) \mathds{1} \{\Phi(x^n,u)=m\},
$$
hence, on average
$
\forall m \in \mathcal{M}_n,$ $\mathbb{E}_{\Phi} \left[ p_{M}(m) \right] =  2^{-nR},
$
which allows us to write
\begin{align}
 &\mathbb{E}_{\Phi} \left[ \mathbf{U}_e \right] \nonumber \\
 & = \mathbb{E}_{\Phi} \left[ \sum_m \left| p_{M}(m) -  \mathbb{E}_{\Phi} \left[  p_{M}(m) \right] \right| \right] \nonumber \\
 &\leq \sum_{i=1}^2  \mathbb{E}_{\Phi} \left[ \sum_m  \left| p^{(i)}_{M}(m) -  \mathbb{E}_{\Phi} \left[  p^{(i)}_{M}(m) \right] \right| \right]  \label{twotermsa},
\end{align}
where $\forall m \in \mathcal{M}_n$, $\forall i \in \llbracket 1,2\rrbracket$,
\begin{align*}
& p^{(i)}_{M}(m) = \sum_{x^n \in \mathcal{A}_i} \sum_u p(x^n,u) \mathds{1} \{\Phi(x^n,u)=m\},
\end{align*}
with $\mathcal{A}_1 \triangleq \mathcal{T}^n_{\epsilon_1}({X})$ and $\mathcal{A}_2 \triangleq\mathcal{A}_1^c .$ After some manipulations similar to those used to bound~\eqref{eq:3}, we bound the second term in~(\ref{twotermsa}) as 
\begin{align}
& \mathbb{E}_{\Phi} \left[ \smash{\sum_m} \left| p^{(2)}_{M}(m) -  \mathbb{E}_{\Phi} \left[  p^{(2)}_{M}(m) \right] \right| \right] 
 \leq 4 |\mathcal{{X}}| e^{-n \epsilon_1^2 \mu_{X}}, \label{ademia}
\end{align}
 \vspace*{-0.0em}
with $\mu_{X}= \displaystyle\min_{x \in \text{supp}(P_{X})}P_{X}(x)$.
Then, we bound the first term in (\ref{twotermsa}) by Jensen's inequality
\begin{multline}
\mathbb{E}_{\Phi} \left[ \sum_m \left| p^{(1)}_{M}(m) -  \mathbb{E}_{\Phi} \left[  p^{(1)}_{M}(m) \right] \right| \right]  \\
\leq \smash{\sum_m }\sqrt{ \textup{Var}_{\Phi}  \left(p^{(1)}_{M}(m) \right)}. \label{eqjensena}
\end{multline}
\vspace*{-0.0em} Moreover, after additional manipulations similar to those used to bound~\eqref{eq:2}, we obtain
\begin{align}
&\textup{Var}_{\Phi}  \left(p^{(1)}_{M}(m) \right) \nonumber  \\ \nonumber
& =  \sum_{x^n \in \mathcal{T}^n_{\epsilon_1}({X})} \sum_u p(x^n,u)^2 \textup{Var}_{\Phi}  \left( \mathds{1} \{ \Phi(x^n,u)=m\} \right)  \\ \nonumber
& \leq  \sum_{x^n \in \mathcal{T}^n_{\epsilon_1}({X})} \sum_u p(x^n,u)^2 \mathbb{E}_{\Phi}  \left[  (\mathds{1} \{ \Phi(x^n,u)=m\})^2 \right]  \\ \nonumber
& =  \sum_{x^n \in \mathcal{T}^n_{\epsilon_1}({X})} \sum_u p(x^n,u)^2 \mathbb{E}_{\Phi}  \left[  \mathds{1} \{ \Phi(x^n,u)=m\} \right]   \\ \nonumber
& =  \sum_{x^n \in \mathcal{T}^n_{\epsilon_1}({X})} \sum_u p(x^n)^2p(u)^2 2^{-nR} \\ \nonumber
& =  \sum_{x^n \in \mathcal{T}^n_{\epsilon_1}({X})} p(x^n)^2 2^{-d} 2^{-nR}\\ \nonumber
& \leq  \sum_{x^n \in \mathcal{T}^n_{\epsilon_1}({X})} \exp_2\left[-2n(1 - \epsilon_1  ){H} ({X}) \right] 2^{-d_n} \frac{1}{M_n}   \\ \nonumber
& \leq  |\mathcal{T}^n_{\epsilon_1}({X})|    \exp_2\left[-2n(1 - \epsilon_1  ){H} ({X}) \right]  2^{-d_n} 2^{-nR} \\ 
& \leq  \exp_2\left[-n (1- 3\epsilon_1  ){H} ({X}) \right]  2^{-d_n} 2^{-nR}. \label{eqjensensuitea} 
\end{align}

\vspace*{-0.em} Thus, by combining (\ref{eqjensena}) and (\ref{eqjensensuitea}), we obtain
\begin{align}
&\mathbb{E}_{\Phi} \left[ \sum_m \left| p^{(1)}_{M}(m) -  \mathbb{E}_{\Phi} \left[  p^{(1)}_{M}(m) \right] \right| \right] \nonumber \\
&\leq \sum_m \sqrt{  \exp_2\left[-n (1- 3\epsilon_1  ){H} ({X}) \right]  2^{-d_n} 2^{-nR} }\\\nonumber
& = \sqrt{M_n} \exp_2\left[-\frac{n}{2} \left((1- 3\epsilon_1  ){H} ({X}) + \frac{d_n}{n} \right) \right] \\
& \leq \exp_2\left[\frac{n}{2} \left(R - (1- 3\epsilon_1  ){H} ({X}) -\frac{d_n}{n}\right) \right] . \label{2demia}
\end{align}
Hence, if $R< {H}({X}) + \frac{d_n}{n} - 3\epsilon_1 {H}({X})$, then asymptotically $\mathbb{E}_{\Phi} \left[\mathbf{U}_e \right] \leq \epsilon$ by (\ref{ademia}) and (\ref{2demia}).
\item We now derive a condition over $R$ to ensure $\mathbb{E}_{\Phi}  [ \mathbf{P}_e ] \leq \epsilon$. We define
$\mathcal{E}_0 \triangleq \{ {X}^n \notin \mathcal{T}_{\epsilon_1}^n({X}) \}$, and $\mathcal{E}_1 \triangleq \{ \exists \hat{{x}}^n \neq {X}^n, \Phi (\hat{{x}}^n,U) = \Phi ({X}^n,U) \text{ and } \hat{{x}}^n \in \mathcal{T}_{\epsilon_1}^n ({X}) \} $
so that by the union bound, $\mathbb{E}_{\Phi}  [ \mathbf{P}_e ] \leq \mathbb{P} [\mathcal{E}_0 ] + \mathbb{P} [\mathcal{E}_1 ].$
We have 
\begin{align}
\mathbb{P} [ \mathcal{E}_0]  \leq 2 |\mathcal{{X}}| e^{-n \epsilon_1^2 \mu_{X}}, \label{eqdemia}
\end{align}
and defining $\mathbf{P}(x^n,\hat{{x}}^n,u) \triangleq \mathbb{P} [ \exists \hat{{x}}^n \neq x^n, \Phi (\hat{{x}}^n,u) = \Phi (x^n,u) \text{ and } \hat{{x}}^n \in \mathcal{T}_{\epsilon_1}^n ({X}) ]$, we have
\begin{align}
\mathbb{P} [\mathcal{E}_1 ]  \nonumber
& = \sum_{x^n} \sum_u p(x^n,u) \mathbf{P}(x^n,\hat{{x}}^n,u) \\ \nonumber
& \leq \sum_{x^n} \sum_u p(x^n,u) \sum_{ \mathclap{ \substack{\hat{{x}}^n \in \mathcal{T}_{\epsilon_1}^n ({X}) \\ \hat{{x}}^n \neq x^n} }} \mathbb{P} [  \Phi(\hat{{x}}^n,u) = \Phi(x^n,u) ]\\ \nonumber
& = \sum_{x^n} \sum_u p(x^n,u) \sum_{ \mathclap{\substack{\hat{{x}}^n \in \mathcal{T}_{\epsilon_1}^n ({X}) \\ \hat{{x}}^n \neq x^n} }} 2^{-nR}\\ \nonumber
& \leq \sum_{x^n} \sum_u p(x^n,u) |\mathcal{T}_{\epsilon_1}^n ({X})| 2^{-nR}\\ \nonumber
& \leq \sum_{x^n} \sum_u p(x^n,u) \exp_2 \left[ n {H}({X})(1 + \epsilon_1)\right] 2^{-nR} \\
& \leq \exp_2 \left[ n ({H}({X})(1+ \epsilon_1)-R )\right]. \label{2demi2a}
\end{align}
Hence,  if $R > {H}({X}) + \epsilon_1{H}({X})$, then asymptotically $\mathbb{E}_{\Phi}  (\mathbf{P}_e)\leq \epsilon$ by (\ref{eqdemia}) and (\ref{2demi2a}).
\end{itemize}
All in all, if $R$ is such that $${H}({X}) + \epsilon_1{H}({X}) < R <{H}({X}) + \frac{d_n}{n} - 3\epsilon_1 {H}({X}),$$ then asymptotically by the selection lemma (e.g. \cite[Lemma 2.2]{Bloch11}), $\mathbb{E}_{\Phi}  [\mathbf{U}_e]\leq \epsilon$ and $\mathbb{E}_{\Phi}  [\mathbf{P}_e]\leq \epsilon$. Thus, we choose $d_n$ such that
$$
4 n \epsilon_1 H({X}) < d_n \leq 4 n \epsilon_1 H({X}) + 1.
$$
We can also choose $\epsilon_1 = \frac{\upsilon_n}{\sqrt{{n}}}$,\footnote{Note that we cannot make $\epsilon_1$ decrease faster because of~(\ref{ademia})~and~(\ref{eqdemia}).} for any $\upsilon_n$ with $\lim_{n \to \infty} \upsilon_n = +\infty$, such that 
$$
4 H({X}) < \frac{d_n}{ \upsilon_n\sqrt{n} } \leq 4 H({X}) +(\sqrt{n}\upsilon_n)^{-1},
$$
which means $d_n = \Theta(\upsilon_n\sqrt{n})$. Finally, by means of the selection lemma applied to $\mathbf{P}_e$ and $\mathbf{U}_e$, there exists a realization of ${\Phi} $ such that $\mathbf{U}_e \leq \epsilon$ and $\mathbf{P}_e \leq \epsilon$.

\subsection{Converse} \label{Secconverse}
We first show that any achievable rate $R$ must satisfy $R \geq {H}(X)$. Assume that $R$ is an achievable rate. We note $M \triangleq \phi_n(X^n,U_{d_n})$. We have
\begin{align*}
nR
& \geq {H}(M) \\
& \geq {I}(X^n;M|U_{d_n})   \\
& = {H}(X^n|U_{d_n}) - {H}(X^n|MU_{d_n})  \\
& \stackrel{(a)}{\geq} {H}(X^n|U_{d_n})  - n \delta(\epsilon) \\
& \stackrel{(b)}{=}    n {H}(X) - n \delta(\epsilon),
\end{align*}
where (a) holds by Fano's inequality and (b) holds by independence of $X^n$ and $U_{d_n}$.

Hence it remains to show an upper bound for the optimal scaling of $d_n$. Recall first the Berry-Ess\'{e}en Theorem.

\begin{theorem}[Berry-Ess\'{e}en Theorem]\label{Thm-BEConv}
Let $\{Z_i\}_{i\in\mathbb{N}}$ be a sequence of i.i.d. random variables with $E[Z_1]=\mu$ and $E[(Z_1-\mu)^2]=\sigma_Z^2>0$ and $E[|Z_1-\mu|^3]=\rho_Z<\infty$. Let  $Y_n = \frac{Z_1+Z_2+\cdots+Z_n-n\mu}{\sigma_Z \sqrt{n}}$. Let $F_n$ denote the cumulative distribution function of $Y_n$. Then, for any $x\in\mathbb{R}$,
\begin{align}
|F_n(x) - \Phi(x)| \leq \frac{\alpha\rho_Z}{\sigma_Z^3\sqrt{n}},
\end{align}
where $\Phi$ is the cumulative distribution function of the standard normal distribution with mean zero and variance 1 and $\alpha$ is a constant that depends only on the distribution of $Z_1$.
\end{theorem}

Using Theorem~\ref{Thm-BEConv}, we show the following.
\begin{lemma}\label{lem-1}
Let $\{X_i\}_{i\in\mathbb{N}}$ be a sequence of i.i.d. random variables with each distributed according to $p_X$  such that 
\begin{alignat*}{3}
H(X)&\, \triangleq-\mathbb{E}[\log p_X(X_1)]&\,<\infty,\\
\sigma^2 &\, \triangleq \mathbb{E}\big[(\log p_X(X_1)+H(X))^2\big]&\,>0,\,\,\,\\
\rho &\, \triangleq \mathbb{E}\big[|\log p_X(X_1)+H(X)|^3\big]&<\,\infty.
\end{alignat*}
Then, there exists an $\alpha>0$ such that for any $a>b>0$,
\begin{align*}
\eta_{a,b}&\triangleq \Big|\mathbb{P}[X^n\in \mathcal{T}_n(a,b)]-\big(\Phi(-b)-\Phi(-a)\big)\Big| \leq \frac{2\alpha\rho}{\sigma^3\sqrt{n}},\\
\eta_{\infty,b}&\triangleq\Big|\mathbb{P}[X^n\in \mathcal{T}_n(\infty,b)]-\Phi(-b)\Big| \leq \frac{\alpha\rho}{\sigma^3\sqrt{n}}.
\end{align*}
where
\begin{align*}
\mathcal{T}_n(a,b) \triangleq \left\{ x^{n} \in\mathcal{X}^{n}:\begin{array}{c}2^{-nH(X)-a\sigma\sqrt{n}} \,<\,p_X(x^{n})\\  2^{-nH(X)-b\sigma\sqrt{n}}\,\geq\, p_X(x^n)\end{array}\right\}.
\end{align*}
\end{lemma}
\begin{proof}
Define  $S_n \triangleq {\frac{nH(X)+\sum_{j=1}^n \log_{2} p_X(X_j) }{\sigma\sqrt{n}}}$. Then,
 \begin{align*}
 \mathbb{P}[X^n\in \mathcal{T}_n(a,b)] 
 &=\mathbb{P}\left[-a<S_n\leq -b\right] \\
 &=\mathbb{P}\left[S_n\leq -b\right]-\mathbb{P}\left[S_n\leq -a\right].
 \end{align*}
 Hence, 
 \begin{align}
\eta_{a,b}&\triangleq\big|\,\mathbb{P}[X^n\in \mathcal{T}_n(a,b)]-\big(\Phi(-b)-\Phi(-a)\big)\big|\notag\\
 & \stackrel{(a)}{\leq} \big|\,\mathbb{P}\left[S_n\leq -b\right]-\Phi(-b)\big|
+\big|\mathbb{P}\left[S_n\leq -a\right]-\Phi(-a)\big|\notag\\ \notag
&\stackrel{(b)}{\leq} \frac{2\alpha\rho}{\sigma^3\sqrt{n}}, 
\end{align}
where $(a)$ holds by the triangle inequality, and $(b)$ holds by Theorem~\ref{Thm-BEConv}. The bound on $\eta_{\infty,b}$ holds similarly.
\end{proof}
We will also make use of the following lemma.
\begin{lemma}\label{lem-UeLbnd}
For any $\phi_n :  \mathcal{{X}}^n \times \mathcal{U}_{d_n} \to \mathcal{M}_n$ and for any $\gamma_n \in ]0, M_n[$,\begin{align*}
\mathbf{U}_e(\phi_{n}) \geq 2\left( \mathbb{P}\left[p_{X^n}(X^{n})>\frac{2^{d_n}}{\gamma_n}\right]-\frac{\gamma_{n}}{M_{n}}\right).
\end{align*}
\end{lemma}

\begin{proof}
Let $\phi_n :  \mathcal{X}^n \times \mathcal{U}_{d_n} \to \mathcal{M}_n$. We apply \cite[Lemma 2.1.2]{HanBook} to $\phi_n$ so that for any $n\in \mathbb{N}$, for any $a$, for any $\Upsilon >0$
\begin{align*}
&\frac{1}{2}\mathbf{U}_e(\phi_{n}) \\
&  = \frac{1}{2} \V{p_{\phi_n(X^n,U_{d_n})} ,p_{U_{M_n}}} \\
& \geq \mathbb{P}[(X^n,U_{d_n}) \notin S'_n(a)] - \mathbb{P}[U_{M_n} \in T_n (a+\Upsilon)] -e^{-n \Upsilon}\\
& = \mathbb{P}[X^n \notin S_n(a-d_n/n)] - \mathbb{P}[U_{M_n} \in T_n (a+\Upsilon)] -e^{-n \Upsilon},
\end{align*}
where
\begin{align*}
&S_n'(a) \\
& \triangleq  \left\{ (x^n,u_{d_n}) \in \mathcal{X}^n\times \mathcal{U}_{d_n} : \frac{1}{n} \log \frac{1}{P_{X^n U_{d_n}}(x^n,u_{d_n})} \geq a\right\}\\
&= \left\{ (x^n,u_{d_n}) \in \mathcal{X}^n\times \mathcal{U}_{d_n} : \frac{1}{n} \log \frac{1}{P_{X^n}(x^n)} \geq a - \frac{d_n}{n}\right\}, 
\end{align*}
\begin{eqnarray*}
& S_n(a) & \triangleq  \left\{ x^n \in \mathcal{X}^n : \frac{1}{n} \log \frac{1}{P_{X^n}(x^n)} \geq a \right\},\\
& T_n(a) & \triangleq  \left\{ u \in \mathcal{U}_{M_n} : \frac{1}{n} \log \frac{1}{P_{U_{M_n}}(u)} < a \right\}.
\end{eqnarray*}
For any $\gamma_n \in ]0, M_n[$, we choose  $\Upsilon \triangleq \frac{1}{n} \log \frac{M_n}{\gamma_n}$ and $a \triangleq \frac{1}{n} \log \gamma_n$, such that $a +\Upsilon = \frac{1}{n} \log M_n$ and $\mathbb{P}[U_{M_n} \in T_n (a+\Upsilon)] =0$. Hence, we obtain
\begin{align*}
\frac{1}{2} \mathbf{U}_e(\phi_{n}) 
& \geq \mathbb{P}[X^n \notin S_n(a-d_n/n)] -e^{-n \Upsilon} \\
& = \mathbb{P} \left[\frac{1}{n} \log \frac{1}{P_{X^n}(x^n)} < a-d_n/n\right] -e^{-n \Upsilon}\\
& = \mathbb{P} \left[\frac{1}{n} \log \frac{1}{P_{X^n}(x^n)} < \frac{1}{n} \log (\gamma_n  2^{-d_n} ) \right] -\frac{\gamma_n}{M_n}.
\end{align*}
\end{proof}

\begin{proposition}[Converse]\label{thm-LLCC}
Let for each $n\in\mathbb{N}$, $\mathcal{C}_n$ be an $(2^{nR},n,2^{d_n})$ uniform compression code $\mathcal{C}_n$ for a \ac{DMS} $(\mathcal{X},p_{X})$ such that
\begin{align*}
\lim_{n\rightarrow\infty} \mathbf{P}_e(\phi_n,\psi_n)= \lim_{n\rightarrow\infty} \mathbf{U}_e(\phi_n) = 0.
\end{align*}
Then, $d_n = \Omega(\sqrt{n})$.
\end{proposition}
\begin{proof}
Since the encoding function $\phi_n$ of $\mathcal{C}_n$ utilizes a seed $U_{d_n}$ taking values in $\llbracket 1, 2^{d_{n}}\rrbracket$ that is independent of the source $X^{n}$, we can find $u^*_{d_n}$ such that
\begin{align}
\mathbb{P}[X^{n} \neq \psi_n(\phi_n(X^n,u^*_{d_n}),u^*_{d_n})] \leq \mathbf{P}_e(\phi_n,\psi_n). \label{eqn-sichoice}
\end{align}
Fix $a>b>0$, and define $\mathcal{L}_n(a,b)$ as 
\begin{align*}
\mathcal{L}_n(a,b) \triangleq \big\{ x^{n} \in\mathcal{T}_{n}(a,b):  x^{n} = \psi_n(\phi_n(x^{n}, u^*_{d_n}),u^*_{d_n})\big\}.
\end{align*}
Note that 
\begin{align}
&\mathbb{P}[X^{n}\in \mathcal{L}_n(a,b)]  \notag \\
&\geq \mathbb{P}[X^{n}\in \mathcal{T}_{n}(a,b)] - \mathbb{P}[X^{n} \neq \psi_n(\phi_n(X^n,u^*_{d_n}),u^*_{d_n})]\notag\\
& \stackrel{(a)}{\geq}  \mathbb{P}[X^{n}\in \mathcal{T}_{n}(a,b)]- \mathbf{P}_e(\phi_n,\psi_n)\notag\\
&\stackrel{(b)}{\geq} \Phi(-b)-\Phi(-a) -\frac{2\alpha\rho}{\sigma^3\sqrt{n}} -  \mathbf{P}_e(\phi_n,\psi_n)\notag\\
& \triangleq \upsilon_n(a,b),\notag
\end{align}
where $(a)$ follows from \eqref{eqn-sichoice}, and $(b)$ holds by Lemma \ref{lem-1} with $\sigma^2$, $\rho$ defined therein. Note that for any $x^{n} \in \mathcal{L}_n(a,b)$, $p_X(x^{n}) \leq  2^{-n H(X)-b \sigma \sqrt{n}}$.
Hence,
\begin{align*}
 \frac{|\mathcal{L}_n(a,b)|}{2^{nH(X)+b\sigma\sqrt{n}}}  & \geq \mathbb{P}[X^{n}\in \mathcal{L}_n(a,b)] \geq \upsilon_n(a,b).
\end{align*}
Since $\mathcal{L}_n(a,b)$ is a subset of source realizations for which the code offers perfect reconstruction (when the seed used is $U_{d_n} = u^*_{d_n}$), we have 
\begin{align}
M_{n} \geq |\mathcal{L}_n(a,b)| \geq \upsilon_n(a,b)\,2^{nH(X)+b\sigma\sqrt{n}}  \label{eqn-M_iBnd}.
\end{align}
We now use Lemma~\ref{lem-UeLbnd} with 
\begin{align*}
\gamma_{n} &\triangleq  \upsilon_n(a,b)\, 2^{nH(X)},\quad M_{n} \geq  \upsilon_n(a,b)\,2^{nH(X)+b\sigma\sqrt{n}} ,
\end{align*}
which yields
 \begin{align}
 \mathbb{P}\left[p_{X^n}(X^{n})>{\textstyle \frac{2^{d_n}}{\gamma_{n}}}\right]&\leq  \frac{\mathbf{U}_e(\phi_{n})}{2}+\frac{\gamma_{n}}{|\mathcal M_{n} |} \nonumber \\
 & \leq \frac{\mathbf{U}_e(\phi_{n})}{2}+ 2^{-b\sigma\sqrt{n}}\label{eqn-taileventbnd1}.
  \end{align}
  From Lemma~\ref{lem-1}, it follows that
  \begin{align}
\mathbb{P}\left[{\textstyle p_{X^n}(X^{n})\leq \frac{2^{d_n}}{\gamma_{n}}}\right] \nonumber
 &=  \mathbb{P}\left[p_{X^n}(X^{n})\leq \frac{2^{d_n}}{{\upsilon_n(a,b)}\,2^{n H(X)}}\right] \nonumber \\
 &\leq \Phi\left( \frac{\log \frac{2^{d_n}}{\upsilon_n(a,b)}}{\sigma\sqrt{n}}\right) + \frac{\alpha\rho}{\sigma^3\sqrt{n}}.\label{eqn-taileventbnd2}
  \end{align}
  Combining \eqref{eqn-taileventbnd1} and \eqref{eqn-taileventbnd2}, we obtain
  \begin{align*}
  \Phi\left( \frac{\log \frac{2^{d_n}}{\upsilon_n(a,b)}}{\sigma\sqrt{n}}\right)\geq \beta_n\triangleq 1-\frac{\mathbf{U}_e(\phi_{n})}{2}-2^{-b\sigma \sqrt{n}}-\frac{\alpha\rho}{\sigma^3\sqrt{n}}.
  \end{align*}
  Rearranging terms and taking appropriate limit, we get 
  \begin{align*}
\lim_{n\rightarrow\infty} \frac{d_{n}}{\sigma\sqrt{n}} &=  \Phi^{-1}\left( \Phi\left(\lim_{n\rightarrow\infty} \frac{d_{n}}{\sigma\sqrt{n}} \right)\right)\\
  &= \Phi^{-1}\left(\lim_{n\rightarrow\infty}  \Phi\left( \frac{d_n}{\sigma\sqrt{n}}  \right)\right)\\
  &\geq  \Phi^{-1}\left(\lim_{n\rightarrow\infty}\beta_n \right) = \Phi^{-1}(1)=\infty,
  \end{align*}
  where in the above arguments, we have used the fact that $\Phi$ is invertible, continuous and increasing.
 \end{proof}

\begin{remark} 
In \cite{Chou13}, we  prove a converse for i.i.d. sources that is stronger than Proposition \ref{thm-LLCC}.
If $d_n=o(\sqrt{n})$, then 
\begin{equation*}
\limsup_{n\rightarrow \infty}  \text{ } \mathbf{P}_e(\phi_n,\psi_n)  + \limsup_{n\rightarrow \infty}  \text{ }  \mathbf{U}_e(\phi_n)
  \geq 1.
\end{equation*}
We however do not need this stronger statement for our purpose.
\end{remark}

\section{Proof of Proposition \ref{Extlossy} } 
\label{sec:extractors_uniform_compression}


Let $\epsilon > 0$, $\delta >0$ and $n \in \mathbb{N}$. Let $t$, $m$, and $d_n$ to be expressed later.
We know from~\cite{Dodis04,Dodis05} that there exists an invertible $(m,d,m,t,\epsilon)$-extractor $\textup{EXT}_0$, such that (\ref{d_exp}) is satisfied. Assume that the emitter and the receiver share a sequence $U_{d_n}$ of $d_n$ uniformly distributed bits. As described in Figure \ref{fig:inv_ext_scheme}, we proceed in two steps to encode $X^n$. First, we perform a compression   of $X^n$ to form $S$ based on $\epsilon_0$-letter typical sequences, $\epsilon_0>0$, we note this operation $\phi_n': \mathcal{X}^n \to \mathcal{M}_n'$, such that $S \triangleq \phi_n'(X^n)$, and we note $\psi_n':  \mathcal{M}_n' \to \mathcal{X}^n$ the inverse operation such that 
\begin{equation} \label{eq1a}
\lim_{n \to \infty} \mathbb{P} [ X^n \neq \psi_n' \circ \phi_n'(X^n)] = 0.
\end{equation}
Note that this compression implies  
\begin{equation} \label{compressratea}
\limsup_{n \to \infty} \frac{1}{n} \log |\mathcal{M}'_n| \leq {H}(X) + \delta.
\end{equation}
Then, we apply the extractor $\textup{EXT}_0$ to $S$ and $U_{d_n}$, to form the encoded message $M = \textup{EXT}_0(S,U_{d_n})$. We define the encoding function $\phi_n: \mathcal{X}^n \times \mathcal{U}_{d_n} \to \mathcal{M}_n$ as 
$$\phi_n(X^n,U_{d_n}) \triangleq M = \textup{EXT}_0(\phi_n'(X^n),U_{d_n}),$$
 and the decoding function $\psi_n:  \mathcal{M}_n  \times \mathcal{U}_{d_n} \to \mathcal{X}^n$ as 
\begin{align} 
 \psi_n(M,U_{d_n})  \notag
 & \triangleq \psi_n'(\textup{EXT}_0^{-1}(M,U_{d_n})) \\ \notag
 &= \psi_n'(S) \\
 & =  \psi_n' \circ \phi_n'(X^n),\label{eq2a}
 \end{align}
  which is possible since $\textup{EXT}_0$ is invertible.
Note that by (\ref{eq1a}), (\ref{eq2a}), we have
\begin{align*} 
&\lim_{n \to \infty} \mathbb{P} [ X^n \neq \psi_n ( \phi_n(X^n,U_{d_n}),U_{d_n})] \\
&= \lim_{n \to \infty} \mathbb{P} [ X^n \neq \psi_n' \circ \phi_n'(X^n)]\\
& = 0,
\end{align*} 
and since the sizes of the first input and output of the extractor are the same, by (\ref{compressratea}), we have
\begin{equation*} 
\limsup_{n \to \infty} \frac{1}{n} \log |\mathcal{M}'_n| \leq {H}(X) + \delta.
\end{equation*}
Moreover,~\cite{Dodis05,Dodis04} also shows that $\mathbf{U}_e \leq \epsilon$. It remains to show that for any $\epsilon_b>0$, we can choose $d_n \triangleq \Theta (n^{1/2 + \epsilon_b})$. Let $\epsilon_0 >0$. As in the proof of Lemma~\ref{lm:compression}, we may show
\begin{align*}
{H}_{\infty} (S) 
&=  - \log (\max p_S (s) ) \\
&\geq  n(1-\epsilon_0){H}(X) - \log \left[1+ \frac{\delta_{\epsilon_0}(n)}{ 1- \delta_{\epsilon_0}(n)} \right].
\end{align*}
We define
\begin{equation} \label{eqtdefa}
t \triangleq n(1-\epsilon_0){H}(X) - \log \left[1+ \frac{\delta_{\epsilon_0}(n)}{ 1- \delta_{\epsilon_0}(n)} \right].
\end{equation}
Thus, since the input size $m$ of the extractor verifies $m \leq \lceil n(1+ \epsilon_0){H}(X) \rceil$, by (\ref{d_exp}) and (\ref{eqtdefa}) we obtain
\begin{align*}
d_n
& \leq  n(1+ \epsilon_0){H}(X) - t + 2 \log [n(1+ \epsilon_0){H}(X) ] \\
& \phantom{---}+ 2 \log \frac{1}{\epsilon} + O(1)\\
& =  2 n \epsilon_0{H}(X) + \log \left[1+ \frac{\delta_{\epsilon_0}(n)}{ 1- \delta_{\epsilon_0}(n)} \right] \\
& \phantom{---}+ 2 \log [n(1+ \epsilon_0){H}(X) ]  + 2 \log \frac{1}{\epsilon} + O(1).
\end{align*}
Then, we choose $\epsilon_0 = \frac{\upsilon_n}{\sqrt{{n}}}$, for any $\upsilon_n$ with $\lim_{n \to \infty} \upsilon_n = +\infty$, such that 
\begin{align*}
& \frac{d_n}{\upsilon_n \sqrt{n} }  \leq 2 H(X) +\frac{2}{\upsilon_n \sqrt{n} }  \log \frac{1}{\epsilon} + O\left(\frac{\log n }{\upsilon_n \sqrt{n} } \right),
\end{align*}
which means $d_n=O(\upsilon_n \sqrt{n})$.

\section{Proof of Proposition \ref{proppolar}} \label{Apppol}
Let $\beta \in ]0,1/2[$. Let $n \in \mathbb{N}$ and $N \triangleq 2^n$. We set $A^N \triangleq X^N G_N$. We define the following sets.
\begin{align*}
\mathcal{V}_X  \triangleq &  \left\{ i\in \llbracket 1,N\rrbracket : {H} \left(A_i|A^{i-1} \right) > 1- \delta_N \right\} , \\
{\mathcal{H}_X} \triangleq  &\left\{ i\in \llbracket 1,N\rrbracket : {H}\left(A_i|A^{i-1}\right) >  \delta_N \right\}.
\end{align*}

These sets cardinalities satisfy the following properties.

\begin{lemma} \label{lemcard}
The sets $\mathcal{H}_X $ and $\mathcal{V}_X$ satisfy 
\begin{enumerate}
\item $\lim_{N \rightarrow + \infty} |\mathcal{H}_X | / N = {H}(X)$,
\item $\lim_{N \rightarrow + \infty} |\mathcal{V}_X | / N = {H}(X)$,
\item $ \lim_{N \rightarrow + \infty}  |\mathcal{H}_X \backslash \mathcal{V}_X| / N  = 0.$
\end{enumerate}
\end{lemma}
\begin{proof}
1) follows from \cite{Arikan10} and \cite{arikan2009rate}. 2) follows from \cite[Lemma 1]{Chou13b} which also uses \cite{arikan2009rate}. 3) holds by 1) and 2) since $\mathcal{V}_X \subset \mathcal{H}_X$.
\end{proof}

\begin{lemma} \label{lemuni}
The output of the encoder $A^N[\mathcal{V}_X]$ is near uniformly distributed with respect to the Kullback-Leibler divergence.
\end{lemma}
\begin{proof} 
We have 
\begin{align*}
{H}\left(A^N[\mathcal{V}_X]\right) =
& \sum_{i \in \mathcal{V}_X} H\left(A_i|A^{i-1}[\mathcal{V}_X]\right) \\
& \geq \sum_{i \in \mathcal{V}_X} {H\left(A_i|A^{i-1}\right)} \\
& \geq |\mathcal{V}_X| \left(1- \delta_N\right), 
\end{align*}
where the first inequality holds because conditioning reduces entropy and the last inequality follows from the definition of $\mathcal{V}_X$. We thus obtain
\begin{align*}
\log 2^{|\mathcal{V}_X|} - H(A^N[\mathcal{V}_X])  \leq   |\mathcal{V}_X| \delta_N
& \leq N \delta_N. \qedhere
\end{align*}
\end{proof}

Finally, by \cite{Arikan10}, the receiver can reconstruct $X^N$ from $A^N[\mathcal{V}_X]$ and $I_0 \triangleq A^N[\mathcal{H}_X \backslash {V}_X]$, where $I_0$ is encrypted via a one-time pad with the uniform seed shared by Alice and Bob. Hence, by Lemmas \ref{lemcard}, \ref{lemuni}, we obtain a polar code construction for a uniform compression code, whose seed length scales as~$o(N)$. 

\bibliographystyle{IEEEtran}
\bibliography{bib,2014Gatech}
\end{document}

%% file: fct_defs.tex
\newcommand{\E}[2][]{{\mathbb{E}_{#1}}{\left[#2\right]}}       
\renewcommand{\P}[2][]{{\mathbb{P}_{#1}}{\left[#2\right]}}     
\newcommand{\V}[1]{{{\mathbb{V}}\left(#1\right)}}              

\newcommand{\card}[1]{\ensuremath{\left|{#1}\right|}}           
\newcommand{\abs}[1]{\ensuremath{\left|#1\right|}}              
\newcommand{\eqdef}{\ensuremath{\triangleq}}                    
\newcommand{\intseq}[2]{\ensuremath{\left\llbracket{#1},{#2}\right\rrbracket}}  

%% file: rv_defs.tex
\DeclareMathAlphabet{\pala}{OT1}{pag}{m}{sl}
\DeclareMathAlphabet{\eurm}{U}{eur}{m}{n}
\DeclareMathAlphabet{\mathbsf}{OT1}{cmss}{bx}{n}
\DeclareMathAlphabet{\mathssf}{OT1}{cmss}{m}{sl}
\DeclareMathAlphabet{\mathcsf}{OT1}{cmss}{sbc}{n}

\DeclareSymbolFont{bsfletters}{OT1}{cmss}{bx}{n}  
\DeclareSymbolFont{ssfletters}{OT1}{cmss}{m}{n}
\DeclareMathSymbol{\bsfGamma}{0}{bsfletters}{'000}
\DeclareMathSymbol{\ssfGamma}{0}{ssfletters}{'000}
\DeclareMathSymbol{\bsfDelta}{0}{bsfletters}{'001}
\DeclareMathSymbol{\ssfDelta}{0}{ssfletters}{'001}
\DeclareMathSymbol{\bsfTheta}{0}{bsfletters}{'002}
\DeclareMathSymbol{\ssfTheta}{0}{ssfletters}{'002}
\DeclareMathSymbol{\bsfLambda}{0}{bsfletters}{'003}
\DeclareMathSymbol{\ssfLambda}{0}{ssfletters}{'003}
\DeclareMathSymbol{\bsfXi}{0}{bsfletters}{'004}
\DeclareMathSymbol{\ssfXi}{0}{ssfletters}{'004}
\DeclareMathSymbol{\bsfPi}{0}{bsfletters}{'005}
\DeclareMathSymbol{\ssfPi}{0}{ssfletters}{'005}
\DeclareMathSymbol{\bsfSigma}{0}{bsfletters}{'006}
\DeclareMathSymbol{\ssfSigma}{0}{ssfletters}{'006}
\DeclareMathSymbol{\bsfUpsilon}{0}{bsfletters}{'007}
\DeclareMathSymbol{\ssfUpsilon}{0}{ssfletters}{'007}
\DeclareMathSymbol{\bsfPhi}{0}{bsfletters}{'010}
\DeclareMathSymbol{\ssfPhi}{0}{ssfletters}{'010}
\DeclareMathSymbol{\bsfPsi}{0}{bsfletters}{'011}
\DeclareMathSymbol{\ssfPsi}{0}{ssfletters}{'011}
\DeclareMathSymbol{\bsfOmega}{0}{bsfletters}{'012}
\DeclareMathSymbol{\ssfOmega}{0}{ssfletters}{'012}

\newcommand{\calC}{{\mathcal{C}}}

\newcommand{\calX}{{\mathcal{X}}}
\newcommand{\calY}{{\mathcal{Y}}}

\newcommand{\calZ}{{\mathcal{Z}}}

%% file: model.pdf_t
\begin{picture}(0,0)%
\includegraphics{model.pdf}%
\end{picture}%
\setlength{\unitlength}{4144sp}%
\begingroup\makeatletter\ifx\SetFigFont\undefined%
\gdef\SetFigFont#1#2#3#4#5{%
  \reset@font\fontsize{#1}{#2pt}%
  \fontfamily{#3}\fontseries{#4}\fontshape{#5}%
  \selectfont}%
\fi\endgroup%
\begin{picture}(5970,1824)(8716,-6148)
\put(8731,-5056){\makebox(0,0)[rb]{\smash{{\SetFigFont{12}{14.4}{\rmdefault}{\mddefault}{\updefault}{\color[rgb]{0,0,0}$V_p^n$}%
}}}}
\put(8731,-4831){\makebox(0,0)[rb]{\smash{{\SetFigFont{12}{14.4}{\rmdefault}{\mddefault}{\updefault}{\color[rgb]{0,0,0}$V_c^n$}%
}}}}
\put(14671,-4831){\makebox(0,0)[lb]{\smash{{\SetFigFont{12}{14.4}{\rmdefault}{\mddefault}{\updefault}{\color[rgb]{0,0,0}$\hat{V}_c^n$}%
}}}}
\put(14671,-5056){\makebox(0,0)[lb]{\smash{{\SetFigFont{12}{14.4}{\rmdefault}{\mddefault}{\updefault}{\color[rgb]{0,0,0}$\hat{V}_p^n$}%
}}}}
\put(10801,-4786){\makebox(0,0)[b]{\smash{{\SetFigFont{12}{14.4}{\rmdefault}{\mddefault}{\updefault}{\color[rgb]{0,0,0}$X^n$}%
}}}}
\put(12601,-4786){\makebox(0,0)[b]{\smash{{\SetFigFont{12}{14.4}{\rmdefault}{\mddefault}{\updefault}{\color[rgb]{0,0,0}$Y^n$}%
}}}}
\put(12871,-5956){\makebox(0,0)[lb]{\smash{{\SetFigFont{12}{14.4}{\rmdefault}{\mddefault}{\updefault}{\color[rgb]{0,0,0}$Z^n$}%
}}}}
\put(11746,-5281){\makebox(0,0)[b]{\smash{{\SetFigFont{12}{14.4}{\rmdefault}{\mddefault}{\updefault}{\color[rgb]{0,0,0}$p_{YZ|X}$}%
}}}}
\put(9901,-4966){\makebox(0,0)[b]{\smash{{\SetFigFont{12}{14.4}{\rmdefault}{\mddefault}{\updefault}{\color[rgb]{0,0,0}ENCODER}%
}}}}
\put(13501,-4966){\makebox(0,0)[b]{\smash{{\SetFigFont{12}{14.4}{\rmdefault}{\mddefault}{\updefault}{\color[rgb]{0,0,0}DECODER}%
}}}}
\end{picture}%

%% file: architecture-1.pdf_t
\begin{picture}(0,0)%
\includegraphics{architecture-1.pdf}%
\end{picture}%
\setlength{\unitlength}{4144sp}%
\begingroup\makeatletter\ifx\SetFigFont\undefined%
\gdef\SetFigFont#1#2#3#4#5{%
  \reset@font\fontsize{#1}{#2pt}%
  \fontfamily{#3}\fontseries{#4}\fontshape{#5}%
  \selectfont}%
\fi\endgroup%
\begin{picture}(9024,3399)(7864,-5923)
\put(8776,-3751){\makebox(0,0)[b]{\smash{{\SetFigFont{12}{14.4}{\rmdefault}{\mddefault}{\updefault}{\color[rgb]{0,0,0}SOURCE}%
}}}}
\put(8776,-3976){\makebox(0,0)[b]{\smash{{\SetFigFont{12}{14.4}{\rmdefault}{\mddefault}{\updefault}{\color[rgb]{0,0,0}ENCODER}%
}}}}
\put(14401,-3121){\makebox(0,0)[b]{\smash{{\SetFigFont{15}{14.4}{\rmdefault}{\mddefault}{\updefault}{\color[rgb]{0,0,0}$\hat{V}_c^n$}%
}}}}
\put(15976,-3076){\makebox(0,0)[b]{\smash{{\SetFigFont{15}{14.4}{\rmdefault}{\mddefault}{\updefault}{\color[rgb]{0,0,0}$\hat{V}_p^n$}%
}}}}
\put(14446,-4606){\makebox(0,0)[lb]{\smash{{\SetFigFont{15}{14.4}{\rmdefault}{\mddefault}{\updefault}{\color[rgb]{0,0,0}$\hat{M}_p$}%
}}}}
\put(16021,-4606){\makebox(0,0)[lb]{\smash{{\SetFigFont{15}{14.4}{\rmdefault}{\mddefault}{\updefault}{\color[rgb]{0,0,0}$\hat{M}_c$}%
}}}}
\put(11251,-5191){\makebox(0,0)[b]{\smash{{\SetFigFont{12}{14.4}{\rmdefault}{\mddefault}{\updefault}{\color[rgb]{0,0,0}MODIFIED}%
}}}}
\put(11251,-5416){\makebox(0,0)[b]{\smash{{\SetFigFont{12}{14.4}{\rmdefault}{\mddefault}{\updefault}{\color[rgb]{0,0,0}WIRETAP}%
}}}}
\put(11251,-5641){\makebox(0,0)[b]{\smash{{\SetFigFont{12}{14.4}{\rmdefault}{\mddefault}{\updefault}{\color[rgb]{0,0,0}ENCODER}%
}}}}
\put(7966,-2716){\makebox(0,0)[lb]{\smash{{\SetFigFont{12}{14.4}{\rmdefault}{\mddefault}{\updefault}{\color[rgb]{0,0,0}\textsc{Application Layer}}%
}}}}
\put(8776,-3076){\makebox(0,0)[b]{\smash{{\SetFigFont{15}{14.4}{\rmdefault}{\mddefault}{\updefault}{\color[rgb]{0,0,0}$V_c^n$}%
}}}}
\put(10306,-4606){\makebox(0,0)[rb]{\smash{{\SetFigFont{15}{14.4}{\rmdefault}{\mddefault}{\updefault}{\color[rgb]{0,0,0}$M_p$}%
}}}}
\put(8731,-4606){\makebox(0,0)[rb]{\smash{{\SetFigFont{15}{14.4}{\rmdefault}{\mddefault}{\updefault}{\color[rgb]{0,0,0}$M_c$}%
}}}}
\put(16831,-2716){\makebox(0,0)[rb]{\smash{{\SetFigFont{12}{14.4}{\rmdefault}{\mddefault}{\updefault}{\color[rgb]{0,0,0}\textsc{Application Layer}}%
}}}}
\put(16831,-5821){\makebox(0,0)[rb]{\smash{{\SetFigFont{12}{14.4}{\rmdefault}{\mddefault}{\updefault}{\color[rgb]{0,0,0}\textsc{Physical Layer}}%
}}}}
\put(10351,-3076){\makebox(0,0)[b]{\smash{{\SetFigFont{15}{14.4}{\rmdefault}{\mddefault}{\updefault}{\color[rgb]{0,0,0}$(U_d,V_p^n)$}%
}}}}
\put(7921,-5821){\makebox(0,0)[lb]{\smash{{\SetFigFont{12}{14.4}{\rmdefault}{\mddefault}{\updefault}{\color[rgb]{0,0,0}\textsc{Physical Layer}}%
}}}}
\put(13501,-5191){\makebox(0,0)[b]{\smash{{\SetFigFont{12}{14.4}{\rmdefault}{\mddefault}{\updefault}{\color[rgb]{0,0,0}MODIFIED}%
}}}}
\put(13501,-5416){\makebox(0,0)[b]{\smash{{\SetFigFont{12}{14.4}{\rmdefault}{\mddefault}{\updefault}{\color[rgb]{0,0,0}WIRETAP}%
}}}}
\put(13501,-5641){\makebox(0,0)[b]{\smash{{\SetFigFont{12}{14.4}{\rmdefault}{\mddefault}{\updefault}{\color[rgb]{0,0,0}DECODER}%
}}}}
\put(14401,-3751){\makebox(0,0)[b]{\smash{{\SetFigFont{12}{14.4}{\rmdefault}{\mddefault}{\updefault}{\color[rgb]{0,0,0}SOURCE}%
}}}}
\put(14401,-3976){\makebox(0,0)[b]{\smash{{\SetFigFont{12}{14.4}{\rmdefault}{\mddefault}{\updefault}{\color[rgb]{0,0,0}DECODER}%
}}}}
\put(15976,-3751){\makebox(0,0)[b]{\smash{{\SetFigFont{12}{14.4}{\rmdefault}{\mddefault}{\updefault}{\color[rgb]{0,0,0}SOURCE}%
}}}}
\put(15976,-3976){\makebox(0,0)[b]{\smash{{\SetFigFont{12}{14.4}{\rmdefault}{\mddefault}{\updefault}{\color[rgb]{0,0,0}DECODER}%
}}}}
\put(10351,-3751){\makebox(0,0)[b]{\smash{{\SetFigFont{12}{14.4}{\rmdefault}{\mddefault}{\updefault}{\color[rgb]{0,0,0}SOURCE}%
}}}}
\put(10351,-3976){\makebox(0,0)[b]{\smash{{\SetFigFont{12}{14.4}{\rmdefault}{\mddefault}{\updefault}{\color[rgb]{0,0,0}ENCODER}%
}}}}
\end{picture}%

%% file: architecture-2.pdf_t
\begin{picture}(0,0)%
\includegraphics{architecture-2.pdf}%
\end{picture}%
\setlength{\unitlength}{4144sp}%
\begingroup\makeatletter\ifx\SetFigFont\undefined%
\gdef\SetFigFont#1#2#3#4#5{%
  \reset@font\fontsize{#1}{#2pt}%
  \fontfamily{#3}\fontseries{#4}\fontshape{#5}%
  \selectfont}%
\fi\endgroup%
\begin{picture}(9024,3399)(7864,-5923)
\put(14401,-3841){\makebox(0,0)[b]{\smash{{\SetFigFont{12}{14.4}{\rmdefault}{\mddefault}{\updefault}{\color[rgb]{0,0,0}SOURCE}%
}}}}
\put(14401,-3121){\makebox(0,0)[b]{\smash{{\SetFigFont{15}{14.4}{\rmdefault}{\mddefault}{\updefault}{\color[rgb]{0,0,0}$\hat{V}_c^n$}%
}}}}
\put(15976,-3076){\makebox(0,0)[b]{\smash{{\SetFigFont{15}{14.4}{\rmdefault}{\mddefault}{\updefault}{\color[rgb]{0,0,0}$\hat{V}_p^n$}%
}}}}
\put(14446,-4606){\makebox(0,0)[lb]{\smash{{\SetFigFont{15}{14.4}{\rmdefault}{\mddefault}{\updefault}{\color[rgb]{0,0,0}$\hat{M}_p$}%
}}}}
\put(16021,-4606){\makebox(0,0)[lb]{\smash{{\SetFigFont{15}{14.4}{\rmdefault}{\mddefault}{\updefault}{\color[rgb]{0,0,0}$\hat{M}_c$}%
}}}}
\put(7966,-2716){\makebox(0,0)[lb]{\smash{{\SetFigFont{12}{14.4}{\rmdefault}{\mddefault}{\updefault}{\color[rgb]{0,0,0}\textsc{Application Layer}}%
}}}}
\put(8776,-3076){\makebox(0,0)[b]{\smash{{\SetFigFont{15}{14.4}{\rmdefault}{\mddefault}{\updefault}{\color[rgb]{0,0,0}$V_c^n$}%
}}}}
\put(10351,-3076){\makebox(0,0)[b]{\smash{{\SetFigFont{15}{14.4}{\rmdefault}{\mddefault}{\updefault}{\color[rgb]{0,0,0}$(U_d,V_p^n)$}%
}}}}
\put(10306,-4606){\makebox(0,0)[rb]{\smash{{\SetFigFont{15}{14.4}{\rmdefault}{\mddefault}{\updefault}{\color[rgb]{0,0,0}$M_p$}%
}}}}
\put(8731,-4606){\makebox(0,0)[rb]{\smash{{\SetFigFont{15}{14.4}{\rmdefault}{\mddefault}{\updefault}{\color[rgb]{0,0,0}$M_c$}%
}}}}
\put(16831,-2716){\makebox(0,0)[rb]{\smash{{\SetFigFont{12}{14.4}{\rmdefault}{\mddefault}{\updefault}{\color[rgb]{0,0,0}\textsc{Application Layer}}%
}}}}
\put(16831,-5821){\makebox(0,0)[rb]{\smash{{\SetFigFont{12}{14.4}{\rmdefault}{\mddefault}{\updefault}{\color[rgb]{0,0,0}\textsc{Physical Layer}}%
}}}}
\put(13501,-5281){\makebox(0,0)[b]{\smash{{\SetFigFont{12}{14.4}{\rmdefault}{\mddefault}{\updefault}{\color[rgb]{0,0,0}WIRETAP}%
}}}}
\put(13501,-5506){\makebox(0,0)[b]{\smash{{\SetFigFont{12}{14.4}{\rmdefault}{\mddefault}{\updefault}{\color[rgb]{0,0,0}DECODER}%
}}}}
\put(11251,-5506){\makebox(0,0)[b]{\smash{{\SetFigFont{12}{14.4}{\rmdefault}{\mddefault}{\updefault}{\color[rgb]{0,0,0}ENCODER}%
}}}}
\put(11251,-5281){\makebox(0,0)[b]{\smash{{\SetFigFont{12}{14.4}{\rmdefault}{\mddefault}{\updefault}{\color[rgb]{0,0,0}WIRETAP}%
}}}}
\put(11566,-3661){\makebox(0,0)[b]{\smash{{\SetFigFont{15}{14.4}{\rmdefault}{\mddefault}{\updefault}{\color[rgb]{0,0,0}$U_d$}%
}}}}
\put(7921,-5821){\makebox(0,0)[lb]{\smash{{\SetFigFont{12}{14.4}{\rmdefault}{\mddefault}{\updefault}{\color[rgb]{0,0,0}\textsc{Physical Layer}}%
}}}}
\put(15976,-3751){\makebox(0,0)[b]{\smash{{\SetFigFont{12}{14.4}{\rmdefault}{\mddefault}{\updefault}{\color[rgb]{0,0,0}SOURCE}%
}}}}
\put(15976,-3976){\makebox(0,0)[b]{\smash{{\SetFigFont{12}{14.4}{\rmdefault}{\mddefault}{\updefault}{\color[rgb]{0,0,0}DECODER}%
}}}}
\put(8776,-3751){\makebox(0,0)[b]{\smash{{\SetFigFont{12}{14.4}{\rmdefault}{\mddefault}{\updefault}{\color[rgb]{0,0,0}SOURCE}%
}}}}
\put(8776,-3976){\makebox(0,0)[b]{\smash{{\SetFigFont{12}{14.4}{\rmdefault}{\mddefault}{\updefault}{\color[rgb]{0,0,0}ENCODER}%
}}}}
\put(10351,-3616){\makebox(0,0)[b]{\smash{{\SetFigFont{12}{14.4}{\rmdefault}{\mddefault}{\updefault}{\color[rgb]{0,0,0}MODIFIED}%
}}}}
\put(10351,-4066){\makebox(0,0)[b]{\smash{{\SetFigFont{12}{14.4}{\rmdefault}{\mddefault}{\updefault}{\color[rgb]{0,0,0}ENCODER}%
}}}}
\put(10351,-3841){\makebox(0,0)[b]{\smash{{\SetFigFont{12}{14.4}{\rmdefault}{\mddefault}{\updefault}{\color[rgb]{0,0,0}SOURCE}%
}}}}
\put(14401,-3616){\makebox(0,0)[b]{\smash{{\SetFigFont{12}{14.4}{\rmdefault}{\mddefault}{\updefault}{\color[rgb]{0,0,0}MODIFIED}%
}}}}
\put(14401,-4066){\makebox(0,0)[b]{\smash{{\SetFigFont{12}{14.4}{\rmdefault}{\mddefault}{\updefault}{\color[rgb]{0,0,0}DECODER}%
}}}}
\end{picture}%

%% file: non-uniform-randomization.pdf_t
\begin{picture}(0,0)%
\includegraphics{non-uniform-randomization.pdf}%
\end{picture}%
\setlength{\unitlength}{4144sp}%
\begingroup\makeatletter\ifx\SetFigFont\undefined%
\gdef\SetFigFont#1#2#3#4#5{%
  \reset@font\fontsize{#1}{#2pt}%
  \fontfamily{#3}\fontseries{#4}\fontshape{#5}%
  \selectfont}%
\fi\endgroup%
\begin{picture}(5970,1824)(8716,-6148)
\put(8731,-5056){\makebox(0,0)[rb]{\smash{{\SetFigFont{15}{14.4}{\rmdefault}{\mddefault}{\updefault}{\color[rgb]{0,0,0}$M_p$}%
}}}}
\put(10801,-4786){\makebox(0,0)[b]{\smash{{\SetFigFont{15}{14.4}{\rmdefault}{\mddefault}{\updefault}{\color[rgb]{0,0,0}$X^n$}%
}}}}
\put(12601,-4786){\makebox(0,0)[b]{\smash{{\SetFigFont{15}{14.4}{\rmdefault}{\mddefault}{\updefault}{\color[rgb]{0,0,0}$Y^n$}%
}}}}
\put(12871,-5956){\makebox(0,0)[lb]{\smash{{\SetFigFont{15}{14.4}{\rmdefault}{\mddefault}{\updefault}{\color[rgb]{0,0,0}$Z^n$}%
}}}}
\put(8731,-4831){\makebox(0,0)[rb]{\smash{{\SetFigFont{15}{14.4}{\rmdefault}{\mddefault}{\updefault}{\color[rgb]{0,0,0}$M_c$}%
}}}}
\put(14671,-4831){\makebox(0,0)[lb]{\smash{{\SetFigFont{15}{14.4}{\rmdefault}{\mddefault}{\updefault}{\color[rgb]{0,0,0}$\hat{M}_c$}%
}}}}
\put(14671,-5056){\makebox(0,0)[lb]{\smash{{\SetFigFont{15}{14.4}{\rmdefault}{\mddefault}{\updefault}{\color[rgb]{0,0,0}$\hat{M}_p$}%
}}}}
\put(11701,-5281){\makebox(0,0)[b]{\smash{{\SetFigFont{15}{14.4}{\rmdefault}{\mddefault}{\updefault}{\color[rgb]{0,0,0}$p_{YZ|X}$}%
}}}}
\put(9901,-4951){\makebox(0,0)[b]{\smash{{\SetFigFont{12}{14.4}{\rmdefault}{\mddefault}{\updefault}{\color[rgb]{0,0,0}ENCODER}%
}}}}
\put(13501,-4944){\makebox(0,0)[b]{\smash{{\SetFigFont{12}{14.4}{\rmdefault}{\mddefault}{\updefault}{\color[rgb]{0,0,0}DECODER}%
}}}}
\end{picture}%